\documentclass[a4paper, USenglish, cleveref, autoref, thm-restate]{lipics-v2021}

\usepackage{tikz}
\usetikzlibrary{automata, positioning, arrows.meta,snakes, calc, patterns, fit, backgrounds}
\usepackage{amssymb}
\usepackage{amsmath}
\usepackage{booktabs}
\usepackage{mathtools}
\usepackage{bm}
\usepackage{pifont}
\usepackage{comment}
\usepackage{fontawesome5}
\usepackage{paralist}

\usepackage{macros}

\title{Population Protocols over Ordered Agents}

\author{Michael {Blondin}}{Département d'informatique, Université de
  Sherbrooke,
  Canada}{michael.blondin@usherbrooke.ca}{https://orcid.org/0000-0003-2914-2734}{supported
  by a Discovery Grant from the Natural Sciences and Engineering
  Research Council of Canada (NSERC).}
\author{Michaël Cadilhac}{DePaul University, Chicago,
  USA}{michael@cadilhac.name}{https://orcid.org/0000-0001-9828-9129}{}
\author{Benjamin {Courchesne}}{Département d'informatique, Université de
  Sherbrooke,
  Canada}{benjamin.courchesne@usherbrooke.ca}{https://orcid.org/0009-0001-8432-4736}{}
\author{Lucie Guillou}{MPI for Software Systems, Kaiserslautern,
  Germany}{lguillou@mpi-sws.org}{https://orcid.org/0000-0002-6101-2895}{}
\author{Corto Mascle}{MPI for Software Systems, Kaiserslautern,
  Germany}{cmascle@mpi-sws.org}{https://orcid.org/0009-0007-7976-7480}{}
\author{Isa Vialard}{MPI for Software Systems, Saarbrücken,
  Germany}{vialard@mpi-sws.org}{https://orcid.org/0000-0002-7261-9342}{}

\authorrunning{M.\ Blondin, M.\ Cadilhac, B.\ Courchesne, L.\ Guillou, C.\ Mascle, I.\ Vialard}
\Copyright{Michael Blondin, Michaël Cadilhac, Benjamin Courchesne, Lucie
  Guillou, Corto Mascle, and Isa Vialard}

\acknowledgements{We thank the anonymous reviewers of ICALP 2026 for their thorough reading and valuable suggestions.}

\nolinenumbers

\hideLIPIcs

\begin{document}

\maketitle

\begin{abstract}
  Population protocols are a distributed computation model in which a collection
  of anonymous, finite-state agents interact in randomly chosen pairs and update
  their states according to a fixed transition function.  The computation is
  defined by the eventual stabilization of the population to a consensus that
  represents the output.  In practice, it is natural to allow each agent to
  carry a unique identifier and compare it with that of another agent before
  interacting.  We model this extension by having agents be totally ordered and
  interactions between two agents to be fireable only if their pair of
  identifiers falls in some condition set.  For instance, \(\PP[<]\) allows for
  two agents to interact only if the first one appears before the second one.

  We study population protocols over ordered agents \(\PP[\cN]\) where \(\cN\) is a
  set of predicates available to restrict transition firing.  We also study
  \(\IOPP[\cN]\), the \emph{immediate observation} fragment of \(\PP[\cN]\) where
  only one agent changes state per interaction.
  Our main result is that \(\IOPP[<]\) recognizes exactly the 
  unambiguous star-free languages, which admits many other characterizations,
  such as two-variable first-order logic or two-way deterministic
  partially-ordered automata.
  We also provide a logic and an automaton model that fits in \(\PP[<]\).  We
  further show that if the successor predicate appears in a set \(\cN\) of
  \(\NSPACE(n)\)-computable predicates, then
  \(\IOPP[\cN] = \PP[\cN] = \NSPACE(n)\).
  Finally, we investigate the problem of deciding whether a given population
  protocol always stabilizes to a consensus.  While this problem is decidable
  for unordered population protocols, we show that this is undecidable already
  for \(\PP[<]\) and \(\IOPP[+1]\), but conditionally decidable for \(\IOPP[<]\).

\end{abstract}

\clearpage
\setcounter{page}{0}
\noindent
Throughout the article, we provide hyperlinks for navigating between
statements in the main text and proofs in the appendix. To go to the
proof of a statement, click ``\faArrowDown'' in the left margin, and
then click ``\faArrowUp'' to go back to the statement in the main
text.  \tableofcontents

\clearpage

\section{Introduction}

Population protocols form a well-established model of distributed computing
where anonymous agents, with very limited individual computational power, work
collectively to achieve a common task~\cite{AADFP06}.  In this model, an input
is scattered among agents that interact pairwise and must take a decision by
reaching a consensus that is stable, that is, agents must eventually all agree
on the output (``consensus'') and stop changing their mind (``stable'').
Population protocols provide a theoretical framework for reasoning about a wide
range of distributed systems, including networks of mobile sensors, chemical
reaction networks, and social networks~\cite{AR09,CDS14,DF01}.

To familiarize the reader with population protocols, we present a classical
protocol for the task of majority voting. A population consists of $n$ agents
(who are not aware of $n$), each carrying a state from a finite set.  Here,
agents start with either $a$ or $b$ as their state. The population aims at
collectively determining whether there are initially more $a$'s than $b$'s. At
each discrete moment, a pair of agents is chosen arbitrarily and their
respective state, from $Q = \{a, b, \underline{a}, \underline{b}\}$, is updated
according to these rules:
\begin{center}
  \begin{tabular}{ccc}
    \toprule
    \textbf{\emph{active to passive}} &
    \textbf{\emph{propagation of winning side}} &
    \textbf{\emph{tiebreaker}} \\
    \midrule
    $a, b \trans{} \underline{a}, \underline{b}$ &
    $a, \underline{b} \trans{} a, \underline{a}$ &
    $\underline{a}, \underline{b} \trans{} \underline{b}, \underline{b}$ \\
    &
    $b, \underline{a} \trans{} b, \underline{b}$
    & \\
    \bottomrule
  \end{tabular}
\end{center}

Since agents are unordered, each rule $p, q \trans{} p', q'$ also stands for
$q, p \trans{} q', p'$. Here are three possible executions of the protocol
starting from three different ``inputs,'' that is, assignments of an initial
state \(a\) or \(b\) to each agent:
\begin{alignat*}{7}
  aaabb
  &\trans{}\ & aa\underline{a}b\underline{b}
  &\trans{}\ & a\underline{a}\underline{a}\underline{b}\underline{b}
  &\trans{}\ & a\underline{a}\underline{a}\underline{a}\underline{b}
  &\trans{}\ & a\underline{a}\underline{a}\underline{a}\underline{a}, \\
  aabbb
  &\trans{}\ & a\underline{a}\underline{b}bb
  &\trans{}\ & a\underline{a}\underline{a}bb
  &\trans{}\ & \underline{a}\underline{a}\underline{a}\underline{b}b
  &\trans{}\ & \underline{a}\underline{a}\underline{b}\underline{b}b
  &\trans{}\ & \underline{a}\underline{b}\underline{b}\underline{b}b,
  &\trans{}\ & \underline{b}\underline{b}\underline{b}\underline{b}b, \\
  aabb
  &\trans{}\ & a\underline{a}\underline{b}b
  &\trans{}\ & \underline{a}\underline{a}\underline{b}\underline{b}
  &\trans{}\ & \underline{a}\underline{b}\underline{b}\underline{b}
  &\trans{}\ & \underline{b}\underline{b}\underline{b}\underline{b}.
\end{alignat*}

Assuming fair scheduling (\eg choosing agents uniformly at random), one can
show that the population stabilizes (almost surely) to the correct outcome: if
there is a majority of $a$'s initially, then agents eventually remain in
$\{a, \underline{a}\}$, otherwise they eventually remain in
$\{b, \underline{b}\}$.

As agents and rules are unordered, each configuration can be seen as a multiset
$\vec{c} \colon Q \to \N$ where $\vec{c}(q)$ indicates the number of agents in
state $q$. It is known that population protocols compute precisely the subsets
of $\N^Q$ that are semilinear~\cite{AAER07}, or, equivalently, that are
definable in Presburger arithmetic (the first-order theory of the naturals with
order and addition). In particular, majority voting amounts to computing the
predicate $\varphi(\vec{c}) = \vec{c}(a) > \vec{c}(b)$.

From a modeling perspective, the fact that agents are unordered is meant to
correspond to a situation in which agents are replicated and anonymous entities,
and hence have no identifiers and are indistinguishable. Yet, it is natural to
allow replicated agents to be totally ordered, \eg they could be devices with
a unique identifier, such as a serial number, stored in read-only memory. In
this context, interactions may depend on the relationships between these
identifiers.

From the perspective of automata theory, this corresponds to considering
population protocols where configurations are \emph{words} rather than
multisets.  Standard population protocols can be seen as computing commutative
properties of words, such as \(|w|_a > |w|_b\), while the word setting
additionally allows for noncommutative properties, such as ``the middle agent
has an \(a\)'' or ``agents strictly alternate between \(a\) and \(b\).''

\paragraph*{Contribution}

Motivated by the above, we propose to study population protocols with
totally-ordered agents.  The class of population protocols so defined,
\(\PP[\cN]\), is parameterized by a set \(\cN\) of predicates over positions that
can be used to restrict transition firings.  Central to our study is the class
\(\PP[<]\), in which a transition \(p, q \xrightarrow{<} r, s\) can only be applied to two
agents in respective states \(p\) and \(q\) if the first agent appears \emph{before}
the second agent.  We consider a well-studied restriction called \emph{immediate
  observation}~\cite{AAER07} where an interaction can only update a single
agent, called the ``observer.''

Our main result, \cref{thm:iopp:da}, establishes that $\IOPP[<]$ has the same
expressive power as unambiguous star-free languages, an important subclass of
regular languages that admits a trove of characterizations (see~\cite{TessonT02}
for a lovely survey on the pervasiveness of this class in automata theory).
Hence, $\IOPP[<]$ is the class of languages captured by these formalisms over
finite words:
\begin{itemize}
\item partially-ordered unambiguous automata~\cite{Montoya25};
\item partially-ordered two-way deterministic automata~\cite{STV01};
\item $\mathsf{LTL}[\mathbf{F}^{-1}, \mathbf{F}]$: linear temporal
  logic with past and future operators~\cite{EVW02};
\item $\FO^2[<]$: the two-variable fragment of first-order logic with
  order~\cite{TW98};
\item $\Delta_2[<]$: the intersection of the $\exists^* \forall^*$ and
  $\forall^* \exists^*$ fragments of first-order logic with order~\cite{PW95};
\item languages recognized by finite monoids from the variety
  $\DA$~\cite{Schu76}.
\end{itemize}

This provides the first characterization of this class in terms of distributed
computing, rather than automata, logic, or algebra.
We will provide in \Cref{ex:median-language} an example of a $\PP[<]$ protocol whose language is not regular, and therefore not expressible by $\IOPP[<]$ protocols.

In addition, we explore systematically the classes induced by our definitions:
\begin{itemize}
\item In \Cref{sec:sem}, we provide a toolbox to study population protocols over
  ordered agents \(\PP[\cN]\), regardless of \(\cN\).  We study protocols that need
  to stabilize only if they reach a positive consensus, which we call
  \emph{semi-deciders}, as opposed to protocols that always stabilize to a
  consensus, dubbed \emph{deciders}.  We also refine the technology of protocols
  with \emph{stabilizing inputs,} studied in \cite{Ras24}, which enables a
  form of composition between protocols.
\item In \Cref{sec:iopp}, we prove the aforementioned characterization of
  \(\IOPP[<]\), and in \Cref{sec:lbppo}, we provide a natural logic and an
  automaton model expressible in \(\PP[<]\).  In this latter section, we fall
  short of showing exact characterizations, but provide conjectures based on our
  new models.
\item In \Cref{sec:beyond}, we explore the expressiveness of \(\IOPP[\cN]\) and
  \(\PP[\cN]\) when the successor predicate is available, that is, when
  transitions can be restricted to fire only if they act on two adjacent agents.
  We show that if all the predicates of \(\cN\) are \(\NSPACE(n)\) computable, then
  \(\IOPP[\cN] = \PP[\cN] = \NSPACE(n)\) --- this is arguably less surprising than
  our main result on \(\IOPP[<]\), as the successor allows for the left-to-right
  propagation of information.
\item Finally, since protocols are only well-behaved when they are deciders, and
  since this property is semantic, we explore in \Cref{sec:decider-dec} whether
  we can check if a given population protocol is a decider.  We thus ask if the
  \emph{syntax} of deciders is decidable --- this is sometimes called the
  \emph{well-specification} problem.  We show that it is undecidable for
  \(\IOPP[+1]\), \(\PP[+1]\), and \(\PP[<]\), and conditionally decidable for
  \(\IOPP[<]\).
\end{itemize}

\paragraph*{Related work}

Our model is closely related to the \emph{community protocols} of Guerraoui and
Ruppert~\cite{GR09}. These are population protocols where each agent has a
unique identifier; each agent can \emph{store} a constant number of identifiers;
and interactions depend on these identifiers but only with respect to their
relative order. The motivation of Guerraoui and Ruppert was to devise an
extension of population protocols, as mild as possible, which would be
fault-tolerant. They proved that community protocols can decide languages from
$\NSPACE(n \log n)$ while tolerating Byzantine failures of a constant number of
agents~\cite{GR09}. The unique identifier of each agent is considered to be
stored in read-only memory, as in real-world low-cost chips, and so exempt from
Byzantine failures. Our model corresponds to community protocols where each
agent has a \emph{single immutable} register initialized to its unique
identifier.

Bournez, Cohen, and Rabie introduce \emph{homonym protocols}, a parameterized
restriction of community protocols, where the $n$ agents have \(f(n)\)
identifiers~\cite{BCR18}. The cases of $f(n) = 1$ and $f(n) = n$ are
respectively population protocols (everyone has the same identifier) and
community protocols (there are as many identifiers as agents).  Identifiers are
from $[0..f(n)-1]$ and agents can compare them with respect to $x < y$, $x = y$,
$x = y + 1$ and $x = 0$.

In~\cite{GGJKS24}, Gańczorz et al.\ introduce \emph{selective population
  protocols} as an extension of population protocols where the state space is
partitioned into finitely many ``groups,'' and where an interaction picks at
random an initiator in state $s$, and then a responder in state $s'$ from the
group of $s$, provided it is nonempty. This is a powerful model that allows for
zero-tests, \ie checking whether no agent holds a certain state. Thus, selective
population protocols are orthogonal to our model. However, the authors dedicate
a section to the median problem: $\bigcup_{n \geq 0}\Sigma^n a \Sigma^n$. They show that if
selective population protocols are extended with the possibility of comparing
keys, then they can solve the median problem in time $\O(\log^4 n)$. The authors
further provide a short proof that, without leveraging the ``selective'' aspect
of their model (\ie groups and zero-tests), any population protocol for the
median problem must work in expected time $\Omega(n)$. This latter setting, only
briefly discussed in~\cite{GGJKS24}, corresponds to our model.

Further extensions include \emph{mediated population
  protocols}~\cite{MCS11}, where communication edges have an internal state;
\emph{population protocols with unordered data}~\cite{BL23}, where the input
alphabet is infinite; and \emph{population protocols
  for graph class identification problems}~\cite{YOI21,AACFJP05}, where agents
aim at determining whether the communication topology satisfies some property.
For other work on immediate-observation protocols,
see~\cite{EGMW18,WK23,EB23,BGKMWW24}.

\section{Preliminaries}\label{sec:prelims}

\subparagraph{Automata and logic.}  We assume some familiarity with formal
languages, automata theory and logic over finite words (\eg see the
textbook~\cite{EB23}).
For a word \(w \in \Sigma^*\), we write \(w[i]\) for the \(i\)-th letter of \(w\), starting at
1, and \(w[i..j]\) for the infix \(w[i]w[i+1]\cdots w[j]\).  For \(\sigma \in \Sigma\), we write
\(|w|_\sigma\) for the number of occurrences of \(\sigma\) in \(w\).

We write $\FO$ to denote first-order logic over words, where quantifiers range
over positions, that is, the set \(\{1, \ldots, |w|\}\) for a given word \(w\), and where
$a(x)$ holds with respect to $w$ iff $w[x] = a$ (see, \eg
\cite[Chap.~8]{EB23} for formal definitions). We write $\FO[<]$ for the
extension of $\FO$ with the numerical predicate \(<\), that allows to test whether $x
< y$ for two positions $x$ and $y$. For example, the sentence $\varphi = (\exists x)(\forall
y)[a(x) \land ((x < y) \rightarrow b(y))]$ describes the language $\Sigma^* a b^*$. By abuse
of notation, $\FO[<]$ stands for both the set of syntactic sentences and for the
class of languages described by these sentences. It is well known that $\FO[<]$
is the class of star-free (regular) languages.

We write $\Sigma_i[<]$ (\resp $\Pi_i[<]$) for the fragment of $\FO[<]$ of sentences in
prenex normal form with $i$ blocks of alternating quantifiers starting with $\exists$
(\resp $\forall$). For example, $\Sigma^* a b^*$ belongs to $\Sigma_2[<]$ due to the form of our example
$\varphi$. We let $\Delta_i[<] = \Sigma_i[<] \cap \Pi_i[<]$.

\subparagraph{Population protocols over ordered agents.}  A population protocol
(PP) describes how a totally-ordered set of finite-state agents interact and
reach a decision about their overall initial states.  Interactions can happen
between any pair of agents, and predicates are used to restrict how transitions
can be taken, based on the position of the agents in the order.

\emph{(Syntax.)} We extend classical PPs to allow for transitions to carry a
test on the \emph{positions} of the totally-ordered agents.  Let
\(\cN \subseteq 2^{\bbN\times \bbN}\) be any set, whose elements we call \emph{numerical
  predicates}: these will be used as the allowed tests on a
transition.\footnote{These are sometimes called \emph{uniform} numerical
  predicates in the literature, to emphasize the fact that they do not depend on
  the \emph{total} number of agents.  A natural predicate that is not uniform is
  \(\mathbf{max}(x)\) which is true if \(x\) is the position of the last agent.
  This technical difference will not impact our results.}  We let
\(\bfsf{true} = \bbN^2\) be the always-true predicate.  The set \(\PP[\cN]\) of PPs
over \(\cN\) is the set of transition systems \((Q, \Sigma, O, \Delta)\) where
\(Q\) is a finite set of \emph{states,} \(\Sigma \subseteq Q\) is a distinguished subset of
\emph{initial} states, \(O\colon Q \to \{\top, \bot\}\) maps each state to an
\emph{opinion}, and
\(\Delta \subseteq Q^2 \times (\cN \cup \{\bfsf{true}\}) \times Q^2\) is a set of \emph{transitions}. An
element of \(\Delta\) is denoted \(q_1, q_2 \xrightarrow{P} q_3, q_4\), expressing,
intuitively, that if two distinct agents meet, the first being in position \(i\)
and state \(q_1\), the second in position \(j\) and state \(q_2\), such that
\((i, j) \in P\), then the first agent changes its state to \(q_3\) and the second to~\(q_4\).
We use $\Sigma$ for the set of initial states as we would like to see such protocols as language acceptors: The initial configuration should be thought of as an input word over alphabet $\Sigma$.

If \(\cN = \{P_1, P_2, \ldots\}\), we write \(\PP[P_1, P_2, \ldots]\) for
\(\PP[\cN]\).  Classical PPs can be seen as the class \(\PP[\emptyset]\), recalling that we
assume that \(\bfsf{true}\) is always available as a numerical predicate.  Our
main interest is in the class \(\PP[<]\), where \(<\) is seen as the set of
pairs \((i, j)\) with \(i < j\), but we will study more expressive
predicates in \Cref{sec:sem,sec:beyond}.

A PP is \emph{immediate-observation} if at most one agent changes state in each
interaction, \ie every transition is of the form $a,b \xrightarrow{P} a,c$
or $a,b \xrightarrow{P} c,b$.  We write \(\IOPP[\cN]\) for the class of
protocols in \(\PP[\cN]\) that are immediate-observation.

\emph{(Semantics.)}  Since our agents are totally ordered, we define the
\emph{configuration} of a system as a word of $Q^+$.  \emph{Initial
  configurations} are words of $\Sigma^+$.
Let \(u\) and \(v\) be two configurations of the same length, we say
that \(u\) \emph{leads to} \(v\), denoted \(u \to v\), if the two
configurations are equal except at potentially two distinct positions
\(i\) and \(j\), and there is a transition
\(u[i], u[j] \xrightarrow{P} v[i], v[j] \in \Delta\)
with \((i, j) \in P\). We let ${\to^*}$ be the reflexive transitive closure of~${\to}$.

\subparagraph{Consensus and stability.}  With \(w\) a configuration, let
\(O(w) \in \{\bot, \top\}\) be the common opinion of all states appearing in
\(w\), if there is one; otherwise \(O(w)\) is undefined.  A configuration \(w\) is a
\emph{\(b\)-consensus} if \(O(w) = b\). It is further \emph{\(b\)-stable} if
$w \to^* v$ implies that $v$ is a $b$-consensus.

\begin{example}\label{ex:ab-star}
  We give an example of a protocol in $\PP[<]$.  Consider the transition system
  $\protocol=(\{a,b,q_\bot\},\{a,b\},O,\Delta)$, with $O(a)=O(b)=\top$,
  $O(q_\bot)=\bot$, and transitions \mbox{$b, a \xrightarrow{<} q_\bot, q_\bot$} and
  \mbox{\(q_\bot, \_ \xrightarrow{\bfsf{true}} q_\bot, q_\bot\)} with ``$\_$'' standing for
  any state.  Every input configuration belongs to $(a+b)^+$ and is therefore a
  $\top$-consensus (all agents output $\top$ initially).  However, such configurations
  need not be $\top$-stable. For instance, starting from $ba$ we can apply the rule
  $b,a\xrightarrow{<} q_\bot,q_\bot$ and obtain $q_\bot q_\bot$, which is
  $\bot$-stable.  In contrast, every configuration in $a^*b^*$ is $\top$-stable: it is
  a $\top$-consensus and no transition is enabled.
\end{example}

\subparagraph{Language of a PP.}

Consider an infinite sequence \(w_0 \to w_1 \to w_2 \to \cdots\) of
configurations, which we call a \emph{run}. We say that it is \emph{fair}
if for every \(w_i\) that appears infinitely often, each configuration
of $\{v \mid w_i \to v\}$ appears infinitely often as well. By induction,
fairness guarantees that each configuration of $\{v \mid w_i
\to^* v\}$ appears infinitely often as well. Intuitively, fairness ensures that reachable
configurations cannot be avoided forever (in a probabilistic setting, where the scheduling induces a probability distribution on the runs, the resulting runs are almost surely fair).  We will
assume that any configuration can be extended into a run, and so into
a fair run, by implicitly adding ``no operation'' transitions.

The \emph{language} of a PP \(\cP\) is the set \(L(\cP)\) of initial configurations
from which there is a fair run that visits a \(\top\)-stable
configuration. Naturally, an initial
configuration can be the origin of a fair run visiting a \(\top\)-stable
configuration, another run visiting a \(\bot\)-stable configuration, or even a run
that visits no stable configurations.  We single out PPs that have more crisp
behaviors.  We say that a PP is a \emph{decider} when for all
\(w \in \Sigma^*\), there is a \(b \in \{\top, \bot\}\) such that \emph{all} fair runs from
\(w\) visit a \(b\)-stable configuration (i.e., for all finite runs from $w$ to some configuration $u$, there is a path from $u$ to a $b$-stable configuration).  
This is usually called
\emph{well-specified} in the literature and we justify our nomenclature in
\Cref{ssec:semidec}.
We say that \(L(\cP)\) is \(\PP[\cN]\)-decidable or \(\IOPP[\cN]\)-decidable
with the obvious meaning.  We will identify \(\PP[\cN]\) and \(\IOPP[\cN]\) with the
class of languages decidable by these protocols. 

Note that since population protocols are ill-defined when no agents are present,
we adopt the convention, when working with their languages, to disregard the
empty word.

\begin{example}\label{ex:ab-star:followup}
  Consider the protocol of \(\PP[<]\) from \Cref{ex:ab-star}.  We show that it
  decides the language $a^*b^*$.  Every input in $a^*b^*$ is already
  $\top$-stable.  Conversely, if an input word is not in $a^*b^*$, then it contains
  the factor $ba$.  Hence, in any fair run, the transition
  $(b,a)\to(q_\bot,q_\bot)$ is eventually executed.  From that point on, every remaining
  agent eventually interacts with a $q_\bot$-agent and is converted to $q_\bot$, so
  the run reaches a configuration in $q_\bot^*$, which is
  $\bot$-stable.
\end{example}

\newcommand{\iL}{\text{\footnotesize\raisebox{-0.25pt}{\faHandPointRight}}}
\newcommand{\iC}{\text{\scriptsize\raisebox{0.25pt}{\faCrown}}}
\newcommand{\iR}{\text{\footnotesize\raisebox{-0.25pt}{\faHandPointLeft}}\hspace{0.1pt}}
\newcommand{\iLt}{\text{\scriptsize\faHandPointRight}}
\newcommand{\iCt}{\text{\tiny\faCrown}}
\newcommand{\iRt}{\text{\scriptsize\faHandPointLeft}}

\begin{example}\label{ex:median-language}
  Consider the \emph{median language} $L \defeq \bigcup_{n \in \N}
  \Sigma^n a \Sigma^n$. Let us describe a protocol $\cP \defeq (Q,
  \Sigma, O, \Delta)$ that decides $L$. The states are defined as $Q
  \defeq \Sigma \times \{\iL, \iC, \iR\} \times \{\bot, \top\}$. The
  components respectively represent the input letter; a belief on
  whether the center is on the right, here, or on the left; and a
  belief on the output.

  We identify input $a$ with state $(a, \iC, \top)$, and each input
  $\sigma \neq a$ with $(\sigma, \iC, \bot)$. We set $O((x, y, z))
  \defeq z$. The set $\Delta$ is defined by these rules, each
  describing a family of transitions:

  \begin{center}
    \begingroup
    \setlength{\tabcolsep}{2pt}
    \begin{tabular}{cp{7pt}rcll}
      \toprule
      && \multicolumn{3}{c}{\textbf{\emph{Population halving}}} \\
      \midrule
      (1) &&
      $(x, \iC, z),\ (x', \iC, z')$ & $\xrightarrow{<}$ &
      $(x, \iL, \bot),\ (x', \iR, \bot)$ \\
      \midrule    
      && \multicolumn{3}{c}{\textbf{\emph{Center finding}}} \\
      \midrule
      (2) &&
      $(x, \iC, z), (x', \iL, z')$ & $\xrightarrow{<}$ &
      $(x, \iL, \bot), (x', \iC, x' = a)$ \\
      (3) &&
      $(x, \iR, z), (x', \iC, z')$ & $\xrightarrow{<}$ &
      $(x, \iC, x = a), (x', \iR, \bot)$ \\
      \midrule
      && \multicolumn{3}{c}{\textbf{\emph{Output propagation}}} \\
      \midrule
      (4) &&
      $(x, y, z), (x', \iC, z')$ & $\xrightarrow{\bfsf{true}}$ &
      $(x, y, z'), (x', \iC, z')$ & for $y \in \{\iL, \iR\}$  \\
      (5) &&
      $(x, y, \top), (x', y', \bot)$ & $\xrightarrow{\bfsf{true}}$ &
      $(x, y, \bot), (x', y', \bot)$ & for $y, y' \in \{\iL, \iR\}$ \\
      \bottomrule
    \end{tabular}
    \endgroup
  \end{center}
  
  By fairness, the first rule must be used until one or zero $\iC$
  remains. Moreover, by fairness, the second and third rules will
  respectively move the $\iL$'s to the left, and the $\iR$'s to the
  right. If some $(x, \iC, y)$ remains, then, by fairness and the
  fourth rule, it will propagate its output $y$, which is $\top$ iff
  $x = a$, by the choice of initial states and by
  rule~(2--3). Otherwise, if the population is of even length, the
  fifth rule will be used to propagate $\bot$.

  \Cref{fig:ex:med} depicts all configurations reachable from the
  initial configuration $aab \in L$. Any \emph{fair} run of a
  population protocol leads to a bottom strongly connected component
  of such a reachability graph. Thus, in this example, every fair runs
  leads to $a_\iLt^\top a_\iCt^\top b_\iRt^\top$, which is
  $\top$-stable. Note that $a_\iLt^\bot a_\iRt^\bot b_\iCt^\bot$ is a
  $\bot$-consensus, but is not $\bot$-stable.

  \begin{figure}[h]
    \centering
    \begin{tikzpicture}[->, semithick, auto, transform shape, scale=0.975]
  \tikzstyle{conf}  = [draw, thin, inner sep=2pt];
  \tikzstyle{confT} = [conf, pattern=north west lines, pattern color=colGood!90!black];
  \tikzstyle{confB} = [conf, pattern=north east lines, pattern color=colBad!40];
  \tikzstyle{confN} = [conf];

  \node[confN] (v0) {$a_\iCt^\top a_\iCt^\top b_\iCt^\bot$};

  \node[confN, below=1.0cm of v0] (v2) {$a_\iLt^\bot a_\iCt^\top b_\iRt^\bot$};
  \node[confB, left =1.5cm of v2] (v1) {$a_\iLt^\bot a_\iRt^\bot b_\iCt^\bot$};
  \node[confN, right=1.5cm of v2] (v3) {$a_\iCt^\top a_\iLt^\bot b_\iRt^\bot$};

  \node[confN, below=1.0cm of v2, xshift=-1.5cm] (v4) {
    $a_\iLt^\top a_\iCt^\top b_\iRt^\bot$
  };
  \node[confN, below=1.0cm of v2, xshift=1.5cm] (v5) {
    $a_\iLt^\bot a_\iCt^\top b_\iRt^\top$
  };

  \node[confT, below=1.0cm of v4, xshift=1.5cm] (vT) {
    $a_\iLt^\top a_\iCt^\top b_\iRt^\top$
  };

  \begin{pgfonlayer}{background}
    \node[fit=(vT), rounded corners=15pt, inner sep=8pt, fill=colGood!15] {};
  \end{pgfonlayer}
  
  \path[->]
  (v0) edge[out=180, in=90, looseness=0.75] node[above] {(1)} (v1)
  (v0) edge node {(1)} (v2)
  (v0) edge[out=0, in=90, looseness=0.75] node[above] {(1)} (v3)

  (v1) edge node {(3)} (v2)

  (v2) edge[bend  left=30] node[swap, yshift=-7pt] {(4)} (v4)
  (v2) edge[bend right=30] node[yshift=-7pt] {(4)} (v5)
  (v4) edge[bend  left=30] node[left] {(5)} (v2)
  (v5) edge[bend right=30] node[right] {(5)} (v2)

  (v3) edge node[above] {(2)} (v2)

  (v4) edge[out=-90, in=160] node[left] {(4)} (vT)
  (v5) edge[out=-90, in=20]  node[right] {(4)} (vT)
  ;
\end{tikzpicture}
    \caption{Configurations reachable from $aab$, where $x_y^z$ stands
      for $(x, y, z)$. Self-loops arising from ``no operation''
      transitions are omitted. The hatched nodes are consensuses; the
      bottom one is stable.}\label{fig:ex:med}
  \end{figure}
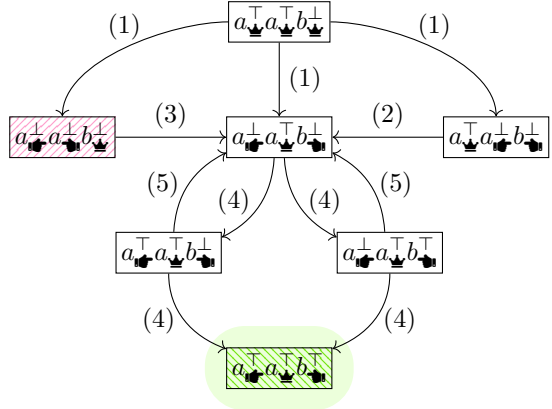
  
  Further observe that the reachability graph of \Cref{fig:ex:med} has
  non-trivial cycles, for instance:
  \[
  a_\iLt^\bot a_\iCt^\top {\color{colA}b_\iRt^\bot} \to
  a_\iLt^\bot a_\iCt^\top {\color{colB}b_\iRt^\top} \to
  a_\iLt^\bot a_\iCt^\top {\color{colA}b_\iRt^\bot} \to
  a_\iLt^\bot a_\iCt^\top {\color{colB}b_\iRt^\top} \to
  \cdots
  \]
  Informally, the first and second agents are fighting to convince the
  third agent. However, by fairness, this is not allowed to happen
  indefinitely. Eventually, the configuration $a_\iLt^\top a_\iCt^\top
  b_\iRt^\top$, at the bottom, is reached.

  In the above specific protocol $\cP$, for each initial
  configuration, the reachability graph has a unique trivial bottom
  strongly connected component, made of one configuration of the form
  $\iL^n \iC\, \iR^n$ or $\iL^n \iR^n$. However, in general, it needs
  not be unique or trivial. A fair run becomes $b$-stable iff it
  visits a bottom strongly connected component whose configurations
  are \emph{all} $b$-consensuses.

  In a decider, for a given initial configuration, \emph{all} bottom
  strongly connected components must consist only of stable
  configurations, all of the \emph{same} output.
\end{example}

We now provide a generic upper bound on the complexity of languages decided by
population protocols; we will exhibit a matching lower bound in
\Cref{sec:beyond} for \(\IOPP[+1]\).  Recall that \(\NSPACE(n)\) is the class of
languages recognized by linear-bounded nondeterministic Turing machines, \ie
nondeterministic machines that require space $O(n)$ over inputs of size $n$.
Languages of this class are exactly the context-sensitive
languages~\cite{Kur64}.
\begin{theorem}\label{thm:ub}
  Let \(\cN\) be a set of numerical predicates, all of which decidable in
  \(\NSPACE(n)\).  We have \(\PP[\cN] \subseteq \NSPACE(n)\).
\end{theorem}
\begin{proof}
Recall that a configuration is $\top$-stable if every reachable configuration from it is a $\top$-consensus (i.e., all agents have opinion $\top$).
  Consider a protocol in \(\PP[\cN]\).  The set of configurations that are
  \emph{not} \(\top\)-stable is decidable in \(\NSPACE(n)\).  Indeed, it is sufficient
  to nondeterministically guess a partial run (\ie a finite sequence of configurations $w_0 \to w_1 \to \dots \to w_n$) that leads to a configuration that
  is not a \(\top\)-consensus.  Each transition can be guessed, its condition
  checked, and its effect applied in \(\NSPACE(n)\).

  By the Immerman--Szelepcsényi theorem, \(\NSPACE(n)\) is closed under complement,
  and so the set of \(\top\)-stable configurations is in \(\NSPACE(n)\).  Since the
  protocol is a decider, to check that a word \(w\) is accepted, it is 
  sufficient to nondeterministically guess a partial run from \(w\), and check
  that it ends in a \(\top\)-stable configuration.  These are all tasks in
  \(\NSPACE(n)\).
\end{proof}

\section{Semantic restrictions of population protocols: a toolbox}\label{sec:sem}

In this section, we define two restrictions of population protocols that will be
used to simplify our constructions.

We first define \emph{semi-deciders}, which do not require the full behavior of
deciders with respect to stable configurations.  Throughout the next sections,
we will see that semi-deciders are much easier to define for some languages, and
we will rely on the forthcoming \Cref{lem:semidec} to combine semi-deciders into
deciders.

We then define \emph{protocols with stabilizing inputs}, which are protocols
where agents can change their mind about their input.  Such protocols are harder
to design, since they are more robust to change, but we show, in
\Cref{lem:renaming}, that they exhibit a strong closure property: they are
closed under alphabet rewriting.

\subsection{Semi-deciders}\label{ssec:semidec}

\begin{plaindef}
  We say that a PP \emph{semi-decides} \(L \subseteq \Sigma^+\) if for all
  \(u \in \Sigma^+\) and every fair run \(\rho\) starting from \(u\), we have
  \(u \in L\) iff \(\rho\) contains a \(\top\)-stable configuration.  We use the terms
  \emph{\(\PP[\cN]\)-semi-decidable} and \emph{\(\IOPP[\cN]\)-semi-decidable} with
  the obvious meaning.
\end{plaindef}

Note that a decider is a semi-decider in which we additionally require that
\(u \notin L\) iff \(\rho\) visits a \(\bot\)-stable configuration; that is, all fair runs
visit a stable configuration.  Illustrating the differences,
\Cref{fig:situations} shows the situations that can occur for runs from an input
\(w\).

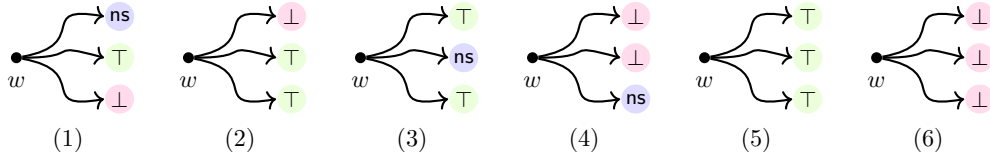
\begin{figure}[h!]
  \centering \makeatletter
\def\mkstable#1#2#3{\@namedef{#1stable}{%
    \tikz[baseline] \node[anchor=base,circle,fill=#2!15,inner sep=1pt, minimum size=12pt] {\makebox[5pt]{\sf #3}};}}
\mkstable{n}{colNeutral}{ns}
\mkstable{t}{colGood}{\(\top\)}
\mkstable{b}{colBad}{\(\bot\)}

\begin{tikzpicture}[
    transform shape, scale=0.91,
  x=1cm,y=1cm,
  nodept/.style={circle, fill=black, inner sep=1.6pt},
  endlab/.style={font=\small},
]

\def\xend{1.5}      
\def\dy{0.6}        
\def\xsep{2.5}      
\def\segl{15pt}     

\node[nodept] (w0) at (0,0) {};
\foreach \k in {1,...,5}{
  \node[nodept] (w\k) at (\k*\xsep,0) {};
}
\foreach \k in {0,...,5}{
  \node[below=2pt of w\k] {$w$};
}

\@namedef{Lab00}{\nstable}
\@namedef{Lab01}{\tstable}
\@namedef{Lab02}{\bstable}

\@namedef{Lab10}{\bstable}
\@namedef{Lab11}{\tstable}
\@namedef{Lab12}{\tstable}

\@namedef{Lab20}{\tstable}
\@namedef{Lab21}{\nstable}
\@namedef{Lab22}{\tstable}

\@namedef{Lab30}{\bstable}
\@namedef{Lab31}{\bstable}
\@namedef{Lab32}{\nstable}

\@namedef{Lab40}{\tstable}
\@namedef{Lab41}{\tstable}
\@namedef{Lab42}{\tstable}

\@namedef{Lab50}{\bstable}
\@namedef{Lab51}{\bstable}
\@namedef{Lab52}{\bstable}

\newcommand{\squig}[4]{%

  \path[thick, ->]
  (#1) edge[out=#3, in=#4, looseness=2] node[] {} ($(#2)-(0.225,0)$)
  ;
}

\foreach \k in {0,...,5}{
  \foreach \i in {0,1,2}{
    \coordinate (E\k-\i) at ($(w\k)+(\xend, {-(\i-1)*\dy})$);
  }

  \node at ($(w\k)+({\xend / 2}, -2*\dy)$) {(\the\numexpr\k+1\relax)};

  \squig{w\k}{E\k-0}{10}{170}
  \squig{w\k}{E\k-1}{0}{165}
  \squig{w\k}{E\k-2}{-10}{190}

  \foreach \i in {0,1,2}{
    \node[endlab]
      at ($(E\k-\i)$)
      {\csname Lab\k\i\endcsname};
  }}

\end{tikzpicture}
\makeatother
    \caption{Configurations labeled {\scriptsize\tstable},
      {\scriptsize\bstable{}} and \(\nstable\) are
      \(\top\)-stable, \(\bot\)-stable, and configurations from which
      no stable configurations are reachable. Arrows depict partial, finite
      runs. Situations (1-6) can happen in PP. Only situations
      (4-6) can happen in semi-deciders, while only situations (5-6)
      can in deciders.}
  \label{fig:situations}
\end{figure}

Our naming convention is justified by the following property:

\begin{restatable}{lemma}{lemSemidec}\labelandarrows{lem:semidec}
  A language $L$ is \(\PP[\cN]\)-decidable
  iff $L$ and its complement are \(\PP[\cN]\)-semi-decidable. The same
  holds for \(\IOPP[\cN]\).
\end{restatable}

\begin{proof}[Proof sketch.]
  We combine the two semi-deciders for $L$ and for $\Sigma^+\setminus L$ by running them in
  parallel.  Each agent stores a pair of states, one for each semi-decider,
  together with a \emph{belief} indicating which semi-decider it currently
  trusts.  Transitions either simulate one step of one of the semi-deciders, or
  \emph{flip} the belief so that beliefs can align; in particular, an agent may
  flip when it meets an agent with the opposite belief, or when its currently
  trusted component produces a $\bot$-witness.  For any input $u$, exactly one of
  the two simulations eventually provides such a witness (since exactly one of
  $u\in L$ or $u\notin L$ holds), which forces all agents to converge to the correct
  belief and yields a stable consensus, hence a decider.
\end{proof}

We note these elementary closure properties:
\begin{restatable}{lemma}{lemInterSemi}\labelandarrows{lem:inter:semi}
  If $L_1, L_2 \subseteq \Sigma^+$ are $\PP[\cN]$-semi-decidable,
  then it is also the case for $L_1 \cap L_2$ and $L_1 \cup L_2$. This further holds
  for $\IOPP[\cN]$ and deciders.
\end{restatable}

\subsection{Protocols with stabilizing inputs}
\label{ssec:stab}

Agents of a protocol are generally not aware that they have reached a stable
consensus and hence ``terminated.'' To carry out a task~A and use its output in
a subsequent task~B, a protocol has to perform \emph{both} tasks concurrently.
The protocol for task~B thus guesses what the output of task~A is going to be,
but ought to be able to self-correct if it becomes clear that the guess was
wrong.  In this subsection, we introduce a formal notion for this
``self-correction'' which will simplify the design of composable protocols;
this is inspired by a recent presentation of~\cite{Ras24} for $\PP[\emptyset]$.

Let us consider protocols where each agent keeps a copy of its
input. Formally, a protocol $\cP = (Q, \Sigma, O, \Delta)$ is said to
be \emph{input-saving} if
\begin{itemize}
\item $Q = \Sigma \times R$ for some finite set $R$;
  
\item Each $\sigma \in \Sigma$ is identified with $(\sigma, r_\sigma)$ for some $r_\sigma \in R$; and
  
\item The first component $\sigma$ of any state $(\sigma, r) \in Q$
  is left unchanged by all transitions of $\Delta$.
\end{itemize}

For all $w \in Q^+$, let $\inpt{w} \in \Sigma^+$ be the projection of $w$ onto its
first component. Given $u, v \in Q^n$, we write $u \dynto v$ if either
$u \to v$, or $v$ equals $u$ except at a single position $i$ where $u[i] =
(\sigma, r)$ and $v[i] = (\sigma', r)$.  This second type of transitions
models a ``change of mind'' of agent $i$ on its input.  We write \(\dynto^*\) for
the reflexive transitive closure of \(\dynto\).  Note what we do not change the
definition of run, which still relies on \(\to\) only.

\begin{plaindef}
  We say that $\cP$ semi-decides $L \subseteq \Sigma^+$ \emph{with stabilizing inputs} if
  (a)~$\cP$ is input-saving, and (b)~for all $u \in \Sigma^+$, all
  $u \dynto^* v$ and every fair run $\rho$ starting from $v$, it is the case that
  $\inpt{v} \in L$ iff $\rho$ visits a $\top$-stable configuration.  We use the term
  ``\emph{decides}'' if (b) is strengthened with the condition that
  $\inpt{v} \notin L$ iff $\rho$ visits a $\bot$-stable configuration.
\end{plaindef}

Note that this is more robust than simply semi-deciding: if \(\cP\) semi-decides
\(L\) with stabilizing inputs, then it semi-decides \(L\).  Intuitively, after
$u \dynto^* v$, the configuration of the population (\ie the projection of $v$
onto the second component) may be \emph{incompatible} with $\inpt{v}$ since
agents have possibly changed their mind several times on their input. Computing
with stabilizing inputs means that the protocol is able to fix its configuration
so that it reflects the expected output on $\inpt{v}$.

It is known that any language $L \in \PP[\emptyset]$ can be decided,
and hence semi-decided, with stabilizing inputs\footnote{It was
claimed earlier by~\cite{AACFJP05}, but definitions and proofs were
deferred to a full paper that never appeared.}~\cite{Ras24}. Let us
turn to an example, which will be useful later:

\begin{restatable}{proposition}{propOrdSemi}\labelandarrows{prop:ord:semi}
  The language $0^* 1^* \cdots k^*$ is
  $\IOPP[<]$-semi-decidable with stabilizing inputs.
\end{restatable}

\begin{proof}[Proof sketch.]
  Let $\Sigma \defeq \{0, 1, \ldots, k\}$ and $L \defeq 0^* 1^* \cdots
  k^*$. The protocol $\cP = (Q, \Sigma, \Delta, O)$ for $L$ is defined
  by $Q \defeq \Sigma \times \{\top, \bot\}$, $O((\sigma, o)) \defeq
  o$, each $\sigma \in \Sigma$ identified with $(\sigma, \top)$, and
  these rules:
  \begin{align*}
    (x, \top), (y, o) &\xrightarrow{\mathmakebox[15pt][c]{<}} (x, \bot), (y, o)
    && \text{for $x > y$}, \\
    (x, \bot), (y, o)
    &\xrightarrow{\mathmakebox[15pt][c]{\bfsf{true}}} (x, \top), (y, o).
  \end{align*}
  The purpose of the first rule is to detect a misordering. The second
  rule allows any agent to nondeterministically reset
  itself.
\end{proof}

We now turn to closure properties of protocols with stabilizing
inputs.  We first cover union and intersection, then move on to
(nondeterministic) alphabet rewriting.

\begin{restatable}{lemma}{lemInterStab}\labelandarrows{lem:inter:stab}
  If $L_1, L_2 \subseteq \Sigma^+$ are $\PP[\cN]$-semi-decidable with stabilizing inputs,
  then it is also the case for $L_1 \cap L_2$ and $L_1 \cup L_2$. This further holds
  for $\IOPP[\cN]$ and deciders.
\end{restatable}

Given $f \colon \Sigma \to 2^\Gamma$ and $w \in \Sigma^n$, let $f(w) \defeq \{w_1' \cdots w_n' \mid w_i' \in
f(w_i) \text{ for each } i \in [1..n]\}$. For example, if $f(0) = \{a, b\}$ and
$f(1) = \{b, c\}$, then $f(01) = \{ab, ac, bb, bc\}$. We extend this notion to
languages: $f(L) \defeq \bigcup_{w \in L} f(w)$.

\begin{restatable}{lemma}{lemRenaming}\labelandarrows{lem:renaming}
  Let $f \colon \Sigma \to 2^\Gamma$. If
  $L \subseteq \Sigma^+$ is $\PP[\cN]$-semi-decidable with stabilizing
  inputs, then $f(L)$ is $\PP[\cN]$-semi-decidable with stabilizing
  inputs. This further holds for $\IOPP[\cN]$.
\end{restatable}
\begin{proof}[Proof sketch.]
  The construction builds a protocol that semi-decides $f(L)$ by
  simulating, in its second component, the input-saving semi-decider
  for $L$ on a \emph{guessed} word $v\in\Sigma^+$ compatible with the
  real input $u\in\Gamma^+$ (\ie $u\in f(v)$).  Agents may revise
  this guess: whenever a $\bot$ opinion is observed in the simulated
  component, an agent is allowed to change its guessed letter to any
  $\sigma$ consistent with its input letter (preserving $u\in
  f(v)$ letterwise).  If $u\in f(L)$, fairness ensures that the
  population can eventually rewrite the guessed word into some $v'\in
  L$, after which the simulation of the semi-decider for $L$ reaches a
  $\top$-stable configuration and changes become disabled; if $u\notin
  f(L)$, reaching a $\top$-stable configuration would force the
  simulated input to lie in $L$, which is not possible.
\end{proof}

\section{Expressiveness of \(\IOPP[<]\)}\label{sec:iopp}

In this section, we provide a precise characterization of the languages decided
by \(\IOPP[<]\).  For any language \(L \subseteq \Sigma^*\), let us write
${\equiv_L}$ for the syntactic congruence of $L$, \ie $w \equiv_L w'$ iff
\(uwv \in L \Leftrightarrow uw'v \in L\) for all $u, v \in \Sigma^*$.  A language is in the class \DA if it
is regular, and satisfies, writing $\alpha(w) \subseteq \Sigma$ for the set of letters appearing
in a word $w$:
\begin{align}
  (\forall w \in \Sigma^*)\left[w \equiv_L w^2 \rightarrow (\forall h \in \alpha(w)^*)[w \equiv_L w\cdot h\cdot w]\right],\label{eq:charac:da}
\end{align}

It is a fascinating result of Pin and Weil~\cite{PW95} that
\(\DA = \Delta_2[<]\).  We will leverage both characterizations to show the theorem
below. The left-to-right inclusion will rely on~\eqref{eq:charac:da}, and the
converse on \(\DA = \Delta_2[<]\), semi-deciders, and protocols with stabilizing
inputs.

\begin{theorem}\label{thm:iopp:da}
  \(\IOPP[<] = \DA\).
\end{theorem}

\subsection{\(\IOPP[<] \subseteq \DA\)}\label{sec:ioppoinda}

We first show that the set of stable configurations in a \(\PP[<]\), and thus in
an \(\IOPP[<]\), admits a simple description.
Given a finite alphabet $\Sigma$ and two words $u,v \in \Sigma^*$, we say that $u$ is a \emph{subword} of $v$, written $u \preceq v$, if $u$ can be obtained from $v$ by removing letters. 
It is well known that ${\preceq}$ is a \emph{well-quasi-order} over $\Sigma^*$~\cite{Higman52}. A consequence is that every strictly decreasing sequence of subword-closed sets must be finite.

\begin{restatable}{lemma}{propBConsensuses}\labelandarrows{lem:b-consensuses}
  Any protocol in \(\PP[<]\) is a
  well-structured transition system w.r.t.~\(\preceq\), that is, for
  all $u, u', w \in Q^*$ with $u \preceq w$ and $u \to u'$, there
  exists $w'$ such that $u' \preceq w'$ and $w \to w'$. Further, the
  set of $b$-stable configurations is subword-closed and computable.
\end{restatable}

As expected from the definition of \DA, the core of the argument showing the
inclusion of \(\IOPP[<]\) in \DA will rely on a pumping argument.  As a first
observation, we show that, in an \(\IOPP[<]\), an infix of the form $wzw$ with
$\alpha(z) \subseteq \alpha(w)$ can, in some sense, mimic the behavior of the protocol on the
simpler infix $w$.

\begin{restatable}{lemma}{ForwardExtension}\labelandarrows{lem:forward-extension}
  In an \(\IOPP[<]\), if $u w v \to^* u' w' v'$ with $|u| = |u'|$ and
  $|v| = |v'|$, then for all $z$ with $\alpha(z)\subseteq \alpha(w)$, there exists
  $z'$ such that $u w z w v \to^* u' w' z' w' v'$ and $\alpha(z') \subseteq \alpha(w')$.
\end{restatable}

We use this first result to show a pumping lemma on $\IOPP[<]$.  It results from
the fact that the sets of $\top$-stable and $\bot$-stable configurations are
subword-closed and \Cref{lem:forward-extension}.

\begin{restatable}{lemma}{lemBigPump}\labelandarrows{bigpump}
  Let \(L \in \IOPP[<]\). There exists a computable $m \geq
  1$ such that, for all $w_1, \ldots, w_m \in \Sigma^+$ and $z \in
  \Sigma^*$ with $\alpha(w_1) = \ldots = \alpha(w_m) \supseteq
  \alpha(z)$, we have $w_1 \cdots w_m \equiv_L (w_1 \cdots w_m) z
  (w_1 \cdots w_m)$.
\end{restatable}

\begin{proof}[Proof sketch.]
\Cref{lem:b-consensuses} indicates that for both $b \in \set{\top,\bot}$, the set of
$b$-stable configurations is subword-closed.  As a consequence, writing
\(B^\eps\) for \(B \cup \{\eps\}\), it is a finite union of languages of the form
$A_1^* B_1^\eps \cdots A_k^* B_k^\eps$ with
$A_i, B_i$ subsets of $Q$, the set of states of our protocol~\cite[Sect.~6.1.1]{Halfon18}.  Let $K_b$ be
the maximal factorization size $k$ over all those languages,
$K = \max(K_\top, K_\bot)$ and $m = 2K+1$.
\Cref{lem:b-consensuses} also states that the set of
$b$-stable configurations is computable for both $b$, in particular we can compute $K_\top, K_\bot$ and $m$.

Intuitively, if we have a run from $w_1 \cdots w_m$ to a $\top$-stable configuration $w' = w'_1 \cdots w'_m$, then we can match $w'$ with one of the expressions above, and some $w'_i$ is entirely contained in some $A_j^*$.

We then use \Cref{lem:forward-extension} to expand the run $w \to w'$ to a run from $w' = w_1 \cdots w_m z w_{1} \cdots w_m$ to some $b$-stable configuration $\widetilde{w'}$, by expanding the $w_i$ section into $w_i \cdots w_m z w_1 \cdots w_i$.

We can thus show that if $w$ is accepted, so is $w'$. Similarly, we show that if $w$ is not accepted, then it can reach a $\bot$-stable configuration, and so can $w'$.
We conclude that one is in the language if and only if the other is. 
To prove the congruence relation we only need to add words $u$ and $v$ around $w$, which does not significantly alter the proof.
\end{proof}

We can now show that \(\IOPP[<]\)-decidable languages satisfy
the definition of \DA.  In \Cref{lem:da:eqs}, we show that
Equation~\ref{eq:charac:da} is satisfied, and in \Cref{lem:reg}, that the
languages are regular, concluding the proof.

\begin{lemma}\label{lem:da:eqs}
	Let \(L \in \IOPP[<]\) and let \(w \in \Sigma^*\) be such that \(w
	\equiv_L w^2\). It is the case that \(w \equiv_L w \cdot h \cdot w\)
	for all \(h \in \alpha(w)^*\).
\end{lemma}
\begin{proof}
	\Cref{bigpump} implies that there is a (computable) number $m$ such that \(w^mhw^m \equiv_L
	w^m\), and since \(w \equiv_L w^2\), we obtain \(w \equiv_L w\cdot h\cdot w\) as claimed.
\end{proof}

\begin{restatable}{lemma}{lemReg}\labelandarrows{lem:reg}
  Every \(\IOPP[<]\)-decidable language is regular. Moreover, given an \(\IOPP[<]\)-protocol $\protocol$, we can construct a finite automaton with the same language.
\end{restatable}
\begin{proof}[Proof sketch.]
  We show that beyond some length, every word must contain a pattern of the form
  $(w_1 \cdots w_m)z(w_1 \cdots w_m)$ with
  $\alpha(w_1) = \ldots = \alpha(w_m) \supseteq \alpha(z)$.  By \Cref{bigpump}, this means that every
  sufficiently long word is $\equiv_L$-equivalent to a shorter one, meaning that
  $\equiv_L$ has finitely many equivalence classes.
  Since we can compute $m$, we can compute this automaton.
\end{proof}

\subsection{\(\DA \subseteq \IOPP[<]\)}\label{ssec:da:iopp}

Recall that $L\in\Pi_2[<]$ iff $\Sigma^*\setminus L \in \Sigma_2[<]$. Since
$\DA$ is equal to $\Delta_2[<] = \Sigma_2[<] \cap \Pi_2[<]$, it is enough to prove that
\(\Sigma_2[<]\) languages are \(\IOPP[<]\)-semi-decidable, appealing to
\Cref{lem:semidec} to conclude.  Note that \(\Sigma_2[<]\) is the set of languages
expressible as finite unions of languages of the form
\(L = A_0^*a_1A_2^*\cdots a_{m-1}A_m^*\) with the \(A_i\)'s being subalphabets (see, \eg
\cite[Thm.~8.8]{Pin97}).  Since semi-deciders are closed under union by
\Cref{lem:inter:semi}, we need only show that \(L\) is
\(\IOPP[<]\)-semi-decidable to conclude.  We start with a technical proposition
that extends a result of~\cite{Ras24} to immediate-observation protocols, then
provide a semi-decider for \(L\) in \Cref{prop:da:semidec}.
Since all constructions used here are effective, our proof also implies that given a finite automaton recognizing a \(\DA\) language, we can effectively construct an \(\IOPP[<]\)-protocol recognizing the same language.

\begin{restatable}{proposition}{propIOEqone}\labelandarrows{prop:io:eqone}
  The language $\{w \in \Sigma^+ \mid |w|_a = 1\}$ is
  $\IOPP[\emptyset]$-semi-decidable with stabilizing inputs.
\end{restatable}

\begin{proof}[Proof sketch.]
  Each agent stores its current input, its last input, and a belief on
  the output. If the current input of an agent mismatches its former
  one, then it resets itself to the current one. Each agent
  eventually stabilizes to a fixed input $(\sigma,\sigma,\cdot)$. The
  belief component then controls the consensus: an $(a,a,\top)$ agent
  can turn any $(\sigma,\sigma,\cdot)$ with $\sigma\neq a$ to $\top$,
  while the presence of any $(\sigma,\sigma,\bot)$ can spread $\bot$
  to other agents.  Finally, if an $a$-agent observes another $a$, it
  switches to $\bot$; combined with the rule that lets an $a$-agent
  reset itself to $\top$, this makes $a$-agents alternate forever
  between $\top$ and $\bot$ whenever there are at least two $a$'s,
  preventing stabilization in that case.  As a result, if the input
  contains no $a$ then every fair run reaches a $\bot$-stable
  consensus; if it contains exactly one $a$ then every fair run
  reaches a $\top$-stable consensus; and if it contains at least two
  $a$'s then some agent flips its belief forever.
\end{proof}

\begin{proposition}\label{prop:da:semidec}
  The language $L \defeq A_0^* a_1 A_2^* \cdots a_{m-1} A_m^*$ is
  $\IOPP[<]$-semi-decidable (with stabilizing inputs).
\end{proposition}

\begin{proof}
  Let $\Gamma \defeq \{0, 1, \ldots, m\}$, $K \defeq 0^* 1^* \cdots
  m^*$ and $K' \defeq \bigcap_{a \in \Gamma, a\text{ odd}} \{w \in
  \Gamma^+ \mid |w|_a = 1\}$. By \Cref{prop:ord:semi}, $K$ is
  $\IOPP[<]$-semi-decidable with stabilizing
  inputs. Furthermore, by \Cref{prop:io:eqone,lem:inter:stab}, $K'$ is
  $\IOPP[\emptyset]$-semi-decidable with stabilizing inputs. By
  \Cref{lem:inter:stab}, $K \cap K'$ is $\IOPP[<]$-semi-decidable with
  stabilizing inputs. Let $f(i) \defeq A_i$ for even \(i\), and
  $f(i) \defeq \{a_i\}$ otherwise. We are done by \Cref{lem:renaming}
  since $L = f(K \cap K')$.
\end{proof}

\section{Expressiveness of \(\PP[<]\)}
\label{sec:lbppo}

The crisp characterization of the previous section ties \(\IOPP[<]\) to a wealth
of computational models with strikingly different flavors.  Chief among them, \DA
is characterized by a logic, \(\Delta_2[<]\), and by partially-ordered unambiguous
automata.  A natural question is thus whether \(\PP[<]\) also admits such a
diverse array of characterizations.  We fall short of providing exact
characterizations, but offer, in this section, two large classes of languages,
one logically-defined and one based on partially-ordered automata, that are
\(\PP[<]\)-decidable.

Our formalisms will involve Presburger arithmetic, the first-order theory of the
naturals with order and addition; \eg
\(\phi(x) = (\exists y \in \N)[y \geq 1 \land x = 2y]\) holds iff $x$ is a positive even number.
Write \(\equiv_c\) for the equivalence of naturals modulo \(c\).  It is well known that
Presburger arithmetic together with \(\equiv_c\) with any \(c \geq 2\) admits quantifier
elimination.  For our purposes, a \emph{Presburger formula} is a (quantifier-free)
Boolean combination of predicates of the form $\sum_{i=1}^n a_i x_i < b$ or
$\sum_{i=1}^n a_i x_i \equiv_c b$, where $a_i, b \in \Z$, $c \geq 2$ and variables
$x_i$ are over $\N$.

\subsection{First-order logic over word intervals}\label{ssec:fo:int}

\def\FOint{\FO^{\mathrm{int}}}%
\def\Piint{\Pi^{\mathrm{int}}}%
\def\Sigmaint{\Sigma^{\mathrm{int}}}%
\def\Deltaint{\Delta^{\mathrm{int}}}%

\begin{plaindef}
  For $w \in \Sigma^n$ and $\sigma \in \Sigma$, let
  $\cnt{\sigma} \colon \{1, \ldots, n\} \times \{1, \ldots, n\} \to \N$ be the function that counts
  the number of occurrences of $\sigma$ between two positions of $w$, \ie
  $\cnt{\sigma}(x, y) \defeq |w[x..y]|_\sigma$.  We define $\FOint$ as first-order logic
  over word \emph{intervals}, that is, with access to numerical values
  \(\cnt{a}(x,y)\) where $x$ and $y$ are either first-order variables, the first
  position (denoted ``$1$'') or the last position (denoted ``$\lmax$'').  We
  will study \(\FOint[<, +, \equiv]\), where we allow numerical values to be compared,
  added, and tested modulo \(c\) for any constant \(c\).  Note that a variable
  \(x\) can be expressed as \(\sum_{\sigma \in \Sigma} \cnt{\sigma}(1, x)\), we will thus assume that
  the atomic formulas only have terms of the form \(\cnt{a}(x, y)\), though we
  write \(1\) for \(\sum_{\sigma\in\Sigma}\cnt\sigma(1, 1)\), and \(\lmax\) for
  \(\sum_{\sigma\in\Sigma}\cnt{\sigma}(1, \lmax)\).  The logics \(\Sigmaint_k\),
  \(\Piint_k\), and \(\Deltaint_k\) are naturally defined.
\end{plaindef}

\begin{example}\label{ex:median}
  Let us show that the \emph{median language} $\bigcup_{n \in \N}
  \Sigma^n a \Sigma^n$ belongs to $\Deltaint_1[<, +, \equiv]$. The
  following predicate asserts that $x$ is the middle position:
  \(\psi(x) \defeq \sum_{\sigma \in \Sigma} \cnt{\sigma}(1, x) =
  \sum_{\sigma \in \Sigma} \cnt{\sigma}(x, \lmax)\).  The median
  language can either be expressed by $(\exists x)[\psi(x) \land
  a(x)]$, or by $(\forall x)[(\psi(x) \rightarrow a(x)) \land
  \lmax \equiv_2 1]$.  Note that $a(x)$ is the same as
  \(\cnt{a}(x, x) = 1\).
\end{example}

We now provide a convenient characterization of languages in $\Sigmaint_1[<, +, \equiv]$.

\begin{restatable}{lemma}{lemSigNormalForm}\labelandarrows{lem:sig:normal:form}
  Any language from $\Sigmaint_1[<, +, \equiv]$ is a finite union of languages of the
  form
  $K = \{a_0 w_1 a_1 \cdots w_m a_m \mid w_i \in \Sigma^*, a_i \in \Sigma, \varphi(x_{i, \sigma} \mapsto |w_i|_\sigma, y_{i,
    \sigma} \mapsto |a_i|_\sigma)\}$ where $m \geq 0$ and $\varphi$ is a Presburger formula over
  variables
  $\{x_{i, \sigma} \mid i \in [1..m], \sigma \in \Sigma\} \cup \{y_{i, \sigma} \mid i \in [0..m], \sigma \in \Sigma\}$.
\end{restatable}

Thanks to the previous lemma, we are able to build a protocol $\PP[<]$ semi-deciding any language in $\Sigmaint_1[<, +, \equiv]$, by building semi-deciders with stabilizing inputs for each $K$, and then using the closure properties presented in \Cref{ssec:stab}.

\begin{proposition}\label{prop:sigma:pp}
  Any language from $\Sigmaint_1[<, +, \equiv]$ is $\PP[<]$-semi-decidable.
\end{proposition}

\begin{proof}
  Let $L \in \Sigmaint_1[<, +, \equiv]$. By
  \Cref{lem:sig:normal:form}, $L$ is a finite union of languages of
  the form $L' = \{a_0 w_1 a_1 \cdots w_m a_m \mid w_i \in \Sigma^*, a_i
  \in \Sigma, \varphi(x_{i, \sigma} \mapsto |w_i|_\sigma, y_{i,
    \sigma} \mapsto |a_i|_\sigma)\}$ where $m \geq 0$ and $\varphi$ is
  a Presburger formula. It suffices to show that $L'$ is
  $\PP[\cN]$-semi-decidable with stabilizing inputs. Indeed, by
  \Cref{lem:inter:stab}, that class of languages is closed under
  union.

  Let $\underline{\Sigma} \defeq \{\underline{\sigma} \mid \sigma \in
  \Sigma\}$, $A_i \defeq \underline{\Sigma} \times \{i\}$, $W_i
  \defeq \Sigma \times \{i\}$ and $\Gamma \defeq \bigcup_{i \in [0..m]} A_i \cup \bigcup_{i \in
  	[1..m]} W_i$. We justify that the three following languages $K$, $K'$ and $K''$ over alphabet $\Gamma$ are semi-decided with stabilizing inputs:
  \begin{itemize}
  	\item $K \defeq A_0^* W_1^* A_1^* \cdots W_m^* A_m^*$,
  	\item $K'\defeq \{w \in \Gamma^+ \mid \varphi(x_{i,\sigma} \mapsto |w|_{(\sigma, i)}, y_{i,\sigma} \mapsto |w|_{(\underline{\sigma}, i)})\}$,
  	\item $K'' \defeq \left\{w \in \Gamma^+ \mid \sum_{\gamma \in A_1} |w|_\gamma = 1 \land \cdots \land \sum_{\gamma \in A_m} |{w}|_\gamma
  	= 1\right\}$   (recall $A_i \cap A_j = \emptyset$ for $i \neq j$).
  \end{itemize}
 Let $f(2i) \defeq A_i$ and $f(2i + 1) \defeq W_{i+1}$. We
  have $f(0^* 1^* 2^* \cdots (2m)^*) = K$.  By
  \Cref{prop:ord:semi,lem:renaming}, $K$ is $\PP[<]$-semi-decidable
  with stabilizing inputs.
  Since $K', K'' \in \PP[\emptyset]$, these two languages are
  $\PP[\emptyset]$-decidable with stabilizing
  inputs~\cite{AACFJP05,Ras24}.

  By \Cref{lem:inter:stab}, the language $K \cap K' \cap K''$ is
  $\PP[<]$-semi-decidable with stabilizing inputs. Let $g((\sigma, i))
  = g((\underline{\sigma}, i)) \defeq \{\sigma\}$.  We are done by
  \Cref{lem:renaming} since $L' = g(K \cap K' \cap K'')$.
\end{proof}

\begin{restatable}{corollary}{corDeltaPP}\labelandarrows{cor:delta:pp}
  $\Deltaint_1[<, +, \equiv] \subseteq \PP[<]$.
\end{restatable}
In the forthcoming \Cref{con:wuPOPA}, we postulate that these two classes coincide.

\subsection{Partially-ordered Parikh automata}\label{ssec:popa}

A \emph{partially-ordered nondeterministic Parikh automaton} (\emph{poPA}) is a
tuple $\A = (Q, \Sigma, \delta, q_0, F, \Psi)$ where:
\begin{compactitem}
\item $Q$ is a finite set of \emph{states} equipped with a partial
  order $\leq$;

\item $\Sigma$ is a finite \emph{alphabet};

\item $\delta \subseteq Q \times \Sigma \times Q$ is the \emph{transition relation}, that satisfies
  $p \leq q$ for all $(p, \sigma, q) \in \delta$; if $p = q$, the transition is dubbed a
  \emph{self-loop}, and otherwise, a \emph{progress transition};

\item $q_0 \in Q$ is the \emph{initial state}, and $F \subseteq Q$ is
  the set of \emph{final states};

\item $\Psi$ is a Presburger formula over variables $\delta$.
\end{compactitem}

\newcommand{\tword}[1]{\pi_\Sigma(#1)}

For all $t = (q, \sigma, q') \in \delta$, we write $q \to^t q'$ and
$\tword{t} \defeq \sigma$. We naturally lift these notations to sequences. A word
$w \in \Sigma^*$ is \emph{accepted} by $\A$ if there exist $\rho \in \delta^*$ and
$q \in F$ such that $q_0 \to^\rho q$, $\tword{\rho} = w$ and
$\Psi(t \mapsto |\rho|_t)$ holds. The \emph{language} of $\A$ is the set $L(\A)$ of words it
accepts. \Cref{fig:automata} depicts two examples of poPA.

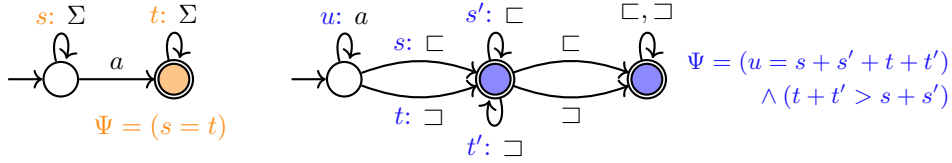
\begin{figure}[h]
  \begin{center}
    \begin{tikzpicture}[->, thick, initial text=, auto]
  \tikzstyle{astate} = [state, minimum size=13pt, inner sep=1pt];
  \newcommand{\ltrans}[1]{{\color{colA}$#1$:}}
  
  \node[astate, initial] (q0) {};
  \node[astate, right=1cm of q0, accepting, fill=colA!50] (q1) {};
  
  \node[below=2pt of q1, xshift=-5pt, colA] {
    $\Psi \defeq (s = t)$
  };
  
  \path[->]
  (q0) edge[loop above] node {\ltrans{s} $\Sigma$} ()
  (q0) edge node {$a$} (q1)
  (q1) edge[loop above] node {\ltrans{t} $\Sigma$} ()
  ;

  \renewcommand{\ltrans}[1]{{\color{colB}$#1$:}}
  
  \node[astate, initial, right=50pt of q1] (q0) {};
  \node[astate, right=1.50cm of q0, accepting, fill=colB!50] (q1) {};
  \node[astate, right=1.50cm of q1, accepting, fill=colB!50] (q2) {};
  
  \node[right=-2pt of q2, colB] {
    \begin{tabular}{r}
      $\Psi \defeq (u = s + s' + t + t')$ \\
      ${} \land (t + t' > s + s')$
    \end{tabular}
  };
  
  \path[->]
  (q0) edge[loop above] node {\ltrans{u} $a$} ()

  (q0) edge[bend  left=20] node[]     {\ltrans{s} $\sqsubset$} (q1)
  (q0) edge[bend right=20] node[swap] {\ltrans{t} $\sqsupset$} (q1)
  
  (q1) edge[loop above] node {\ltrans{s'} $\sqsubset$} ()
  (q1) edge[loop below] node {\ltrans{t'} $\sqsupset$} ()

  (q1) edge[bend  left=20] node[]     {$\sqsubset$} (q2)
  (q1) edge[bend right=20] node[swap] {$\sqsupset$} (q2)
  
  (q2) edge[loop above] node {$\sqsubset, \sqsupset$} ()
  ;
\end{tikzpicture}
    \caption{Example of poPA for the median language $\bigcup_{n
        \geq 0} \Sigma^n a \Sigma^n$ (left) and the coDyck-witness
      language $\{a^n v \mid n \geq 1, v \in \{\sqsubset,
      \sqsupset\}^{\geq n}, |v[1..n]|_\sqsupset >
      |v[1..n]|_\sqsubset\}$ (right).}\label{fig:automata}
  \end{center}
\end{figure}

\Cref{lem:sig:normal:form} allows us to translate a formula from
$\Sigmaint_1[<, +, \equiv]$ into a poPA by guessing a
factorization $a_0 w_1 a_1 \cdots w_m a_m$ with a progress transition
for each $a_i$, and a self-loop for each $w_i$; and then using the
Presbuger acceptance formula of the automaton to verify the
guess. We can also translate a poPA into a formula of
$\Sigmaint_1[<, +, \equiv]$ by guessing the position of the progress
transitions and verifying the validity of the path. This yields:

\begin{restatable}{theorem}{propSigPopa}\labelandarrows{prop:sig:popa}
  A language is recognized by some
  poPA iff it belongs to $\Sigmaint_1[<, +, \equiv]$.
\end{restatable}

Recall that \(\Sigma_2[<]\) is characterized by partially-ordered automata, while
\(\Delta_2[<]\) is characterized by \emph{unambiguous}
partially-ordered automata.  We explore a similar notion for
\(\Deltaint_1[<, +, \equiv]\) and poPA.  A poPA $\A$ is \emph{weakly
  unambiguous}\footnote{The nomenclature, introduced in~\cite{BostanCKN20},
  stems from prior studies~\cite{CadilhacFM13} which called ``unambiguous'' the
  PA with an unambiguous underlying automaton.  We note that the examples of
  \Cref{fig:automata} can be shown, using the tools of~\cite{CadilhacFM13}, not
  to be expressible with unambiguous PA.} if every $w \in L(\A)$ is accepted by at
most one path (w.r.t.\ $F$ \emph{and} $\Psi$).  Note that the automata
of~\Cref{fig:automata} are weakly unambiguous and can be complemented.  By
\Cref{prop:sig:popa}, this means that both languages belong to
$\Deltaint_1[<, +, \equiv]$.  More generally, we conjecture that weakly-unambiguous
poPA are closed under complement.

\begin{restatable}{observation}{obsPaComp}\labelandarrows{obs:pa:comp}
  If weakly-unambiguous poPA are closed under complement, then any
  language recognized by a weakly-unambiguous poPA belongs to
  $\Deltaint_1[<, +, \equiv]$.
\end{restatable}

In the forthcoming \Cref{con:wuPOPA}, we postulate that these two classes coincide.  We note, as
a sanity check, that since \(\Delta_2[<]\) is characterized by unambiguous
partially-ordered automata,
\(\IOPP[<] = \DA = \Delta_2[<] \subsetneq \Deltaint_1[<, +, \equiv]\).  The strictness of the
inclusion is in particular witnessed by the \emph{regular} language
$\{a^{2n} \mid n \geq 1\}$ over alphabet $\{a\}$, which trivially belongs to
$\Deltaint_1[<, +, \equiv]$ with formula \(\lmax \equiv_2 0\).

\begin{rmk}[Two-way models]
  Recall that \DA is also characterized by two-way \emph{deterministic}
  partially-ordered automata.  We can show that two-way poPA are equivalent to
  poPA. However, defining two-way deterministic PA as two-way deterministic
  automata with a Presburger constraint, it is known that the model is as
  expressive as unambiguous PA~\cite{FiliotGM19}, which cannot express the
  languages of \Cref{fig:automata}.
\end{rmk}

\subsection{The regular languages of \(\PP[<]\)}\label{ssec:regpp}

\Cref{lem:sig:normal:form} gives a clear form for the languages of
\(\Sigmaint_1[<, +, \equiv]\).  If the language \(K\) therein is regular, this points to
the formula \(\varphi\) not arithmetically linking the \(w_i\)'s between one another
nor imposing nonregular constraints on any single \(w_i\).  Hence we have
naturally:
\begin{conjecture}\label{conj:pcom}
  Let \(\com\) be the set of commutative regular languages.  Let \(\pcom\) be the
  set of languages that are finite unions of languages
  \(L_0a_1L_1\cdots a_nL_n\) with \(L_i \in \com\) and \(a_i\) letters.  We conjecture that
  the regular languages of \(\Sigmaint_1[<, +, \equiv]\) are exactly \(\pcom\).
\end{conjecture}

The class \(\pcom\) has been studied in previous works: \cite{Anil08} conjectures
that they correspond to the regular languages requiring \(O(\log n)\)
communication and \cite{Cano13} shows that it is included in the largest class
of languages, closed under the so-called positive variety operations, that does
\emph{not} contain \((ab)^+\).  We recall, for contrast, that \DA is the largest
class of \emph{aperiodic} regular languages, closed under the variety
operations, that does not contain \((ab)^+\).  To give additional credence to our
conjecture, we show:
\begin{restatable}{proposition}{propNotAb}\labelandarrows{prop:not:ab}
  The language $(ab)^+$ does not belong to $\Sigmaint_1[<, +, \equiv]$.
\end{restatable}

The class \(\upcom\) is defined similarly as \(\pcom\), except that we require that
every word \(w\) in the marked concatenation has a \emph{unique} decomposition
\(w = w_0a_1w_1\cdots a_nw_n\) with \(w_i \in L_i\).  It is known~\cite{PlaceZeitoun24} that
\(\upcom = \pcom \cap \overline{\pcom}\).  We naturally postulate, in 
\Cref{con:wuPOPA}, that the regular languages of \(\PP[<]\) are exactly \(\upcom\).

\section{Expressiveness of \(\IOPP[\cN]\) and \(\PP[\cN]\) when successor is available}\label{sec:beyond}

We now consider protocols where transitions may be fired only when two agents
are adjacent. Formally, we look at \(\IOPP[\cN]\) and \(\PP[\cN]\) with
\(+1 \in \cN\), where \(+1\) (sometimes written \(\mathbf{succ}\) in the literature) is
the predicate \(\{(x, x+1) \mid x \in \bbN\}\). This study is reminiscent of the
classic study of \(\FO^2[<, +1]\)~\cite{EVW02,LPS10,TW98,WI07} that followed that
of \(\FO^2[<]\), the two-variable fragment of \(\FO[<]\) that is as expressive as
\(\IOPP[<]\).

\begin{example}
 \newcommand{\mArr}[1]{\xrightarrow{\mathmakebox[15pt][c]{\mathclap{#1}}}}
 \newcommand{\qaR}{a^\text{\raisebox{-1pt}{\faAngleRight}}}
 \newcommand{\qbL}{b^\text{\raisebox{-1pt}{\faAngleLeft}}}
 \newcommand{\qaS}{a^\text{\tiny\faHourglassHalf}}
 \newcommand{\qbS}{b^\text{\tiny\faHourglassHalf}}
 \newcommand{\qT}{\text{\scriptsize\faCheck}}
 \newcommand{\qF}{\text{\scriptsize\,\faTimes\,}}
 Recall that the set of all $b$-stable configurations
 of a $\PP[<]$ is subword-closed, and hence regular. The following
 $\PP[+1]$ has a non-regular set of \(\top\)-stable
 configurations:
 \begin{align*}
   \qaS, \qbS &\mArr{\bfsf{true}} \qaR, \qbL &
   \qaR, a  &\mArr{+1} \qT, \qaS &
   \qaR, \qbL &\mArr{+1} \qF, \qF \\
   &&
   b, \qbL &\mArr{+1} \qbS, \qT
 \end{align*}

 Let us set the opinion of each state to $\top$, except for $\qF$. 
 The protocol is constructed so that from a configuration of the form $\qaS a^* b^* \qbS$, one marker moves to the right along the $a$'s, while the other synchronously moves to the left along the $b$'s. If they meet at the frontier between $a$'s and $b$'s, it means there are as many $a$'s and $b$'s, and that state $\qF$ appears. 
 
 The
 set of \(\top\)-stable configurations cannot be regular since, from
 $w \in \qaS a^* b^* \qbS$, the state $\qF$ will appear iff $|w|_a =
 |w|_b$, \eg
 \[
 \qaS a b \qbS \to
 \qaR a b \qbL \to
 \qT \qaS b \qbL \to
 \qT \qaS \qbS \qT \to
 \qT \qaR \qbL \qT \to
 \qT \qF \qF \qT.
 \]
\end{example}

The main result of this section is: 
\begin{theorem} Let \(\cN\) be a set of \(\NSPACE(n)\)-decidable numerical predicates that contains \(+1\). We have \(\IOPP[\cN] = \PP[\cN] = \NSPACE(n)\).
 \end{theorem} 
 To show this result, it is enough to prove that
 \(\NSPACE(n) \subseteq \IOPP[+1]\), since \(\PP[\cN] \subseteq \NSPACE(n)\) has been established
 beforehand in \Cref{thm:ub}.

\subsection{$\NSPACE(n) \subseteq \IOPP[+1]$}
\label{ssec:iopp-pp-succ} 
\label{ssec:nspace-in-pp-succ}

We start by showing that transitions using the successor predicate can
be simulated with immediate-observation transitions.

\begin{restatable}{lemma}{lemPPtoIOPPsucc}\labelandarrows{lem:PP-to-IOPP-succ}
  Let $L$ be a $\PP[+1]$-semi-decidable language. There exists a
  $\PP[+1]$-protocol $\cP'$ semi-deciding $L$ and such that each
  $+1$-transition of $\cP'$ is immediate-observation.
\end{restatable}

\begin{proof}[Proof sketch.]
  The construction turns a $\PP[+1]$-protocol $\protocol$ into a new $\PP[+1]$-protocol $\protocol'$ by replacing each transition
  $\delta = (a,b)\xrightarrow{+1}(c,d)$ with a short \emph{four-step handshake} along
  the successor relation, using auxiliary markers $a^\delta$ and
  $b^{\mathsf{ack},\delta}$. The four rules are:
  \begin{align*}
    a,b &\xrightarrow{+1} a^\delta,b & 
                                  a^\delta,b &\xrightarrow{+1}  a^\delta,b^{\mathsf{ack},\delta}\\
    a^\delta,b^{\mathsf{ack},\delta} &\xrightarrow{+1}  c,b^{\mathsf{ack},\delta} &
                                                                     c,b^{\mathsf{ack},\delta} &\xrightarrow{+1}  c,d
  \end{align*}
  Intuitively, one agent announces that it wants to perform $\delta$, the neighbor
  acknowledges, and the two agents then update one after the other; additional
  clean-up rules erase incomplete handshakes so that executions cannot block
  when several handshakes overlap.
\end{proof}

\newcommand{\Lmk}[1]{\underline{#1}}   
\newcommand{\Rmk}[1]{\overline{#1}}    

\begin{restatable}{lemma}{lemPPsuccSemiDecidesLTM}\labelandarrows{lem:PP-succ-SemiDecides-LTM}
  The language of a linear-bounded nondeterministic Turing machine is
  $\IOPP[+1]$-semi-decidable.
\end{restatable}

\begin{proof}[Proof sketch.]
  Given a linear-bounded nondeterministic Turing machine $M$, we
  construct an $\IOPP[+1]$-protocol $\cP$ that simulates $M$ under
  stabilizing inputs. Thanks to \Cref{lem:PP-to-IOPP-succ}, we only
  need to ensure that $\bfsf{true}$-transitions are
  immediate-observation.

  Let $\Sigma'$ be the tape alphabet of $M$. The protocol $\cP$ works
  over $\Gamma^*$ where $\Gamma \defeq (P \cup \{-\}) \times \Sigma'
  \times \{\text{fst}, -, \text{lst}\}$. For example, $w = (-, a,
  \text{fst}) (p, a, -) (-, b, \text{lst})$ represents the
  configuration of $M$ in state $p$, with the head on the second cell,
  and the tape containing $aab$.

  The protocol $\cP$ simulates $M$ using $+1$-transitions with the
  agent containing the current state $p$. If $M$ accepts, then $\cP$
  spreads $\top$ across the population. There are two challenges:
  \begin{enumerate}
  \item Since we want to semi-decide under \emph{stabilizing inputs},
    $\cP$ must reset the population and the simulation of $M$ whenever
    an agent changes its mind on its input. This is implemented as
    follows. Any agent who has changed its mind can flag the last
    agent. Using $+1$-transitions, this can spread from right to left,
    resetting the population along the way.

  \item If the input of $\cP$ is an invalid configuration of $M$, then
    the population should not reach a stable $\top$-consensus. For
    example, if two agents both represent the head of $M$, then they
    change their belief to $\bot$ upon meeting. Similarly, we must
    detect whether ``fst'' and ``lst'' occur exactly in the first and
    last agent.
  \end{enumerate}

  The resulting protocol $\cP$ does not semi-decide the language of
  $M$, but rather the set of configurations of $M$ leading to
  acceptance. So, we need to intersect with the set of initial
  configurations, and then project onto the second component of
  $\Gamma$. This can be done thanks to
  \Cref{prop:ord:semi,prop:io:eqone,lem:inter:stab,lem:renaming}.
\end{proof}

\begin{corollary}\label{prop:NSPACE-IOPP-succ}
  $\NSPACE(n) \subseteq \IOPP[+1]$.
\end{corollary}

\begin{proof}
  Let $L\in \NSPACE(n)$. As $\NSPACE(n)$ is closed under complement, there
  exist two linear-bounded nondeterministic Turing machines $M$ and
  $\overline{M}$ recognizing $L$ and $\Sigma^+ \setminus L$, respectively. By \Cref{lem:PP-succ-SemiDecides-LTM}, we can build two \(\IOPP[+1]\)-protocols, say $\protocol_M$
  and $\protocol_{\overline{M}}$, that semi-decide $L$ and $\Sigma^+ \setminus L$,
  respectively. \Cref{lem:semidec} allows to conclude.
\end{proof}

\section{Decidability of checking whether a population protocol is a decider}
\label{sec:decider-dec}

A good portion of our results assumes that a given population protocol is a
decider or a semi-decider in order to \emph{construct} an object (a Turing
machine in \Cref{thm:ub}, a formula in \Cref{thm:iopp:da}, etc.).  To fully
understand how constructive these proofs are, we ought to study whether it is
decidable, given a population protocol, to check if it is a decider. In other
words, \emph{is the syntax of deciders decidable?} From the perspective of formal verification, this is also a natural problem, known as the well-specification problem. It is decidable for $\PP[\emptyset]$~\cite{EGLM15}, and undecidable when protocols are extended with an infinite alphabet (where agents carry data from an infinite domain, and transitions can use equality constraints between those data), except for immediate-observation protocols~\cite{BGKMWW24}.

We show that the problem is
undecidable already for \(\PP[<]\) and \(\IOPP[+1]\), and provide a natural conjecture that would entail that
the problem is decidable for \(\IOPP[<]\).

\subsection{The syntax of \(\PP[<]\), $\IOPP[+1]$ and $\PP[+1]$ deciders are undecidable}\label{ssec:syntax:undec}

First, we consider the \emph{emptiness problem}: given a
\(\PP[\mathcal{N}]\) protocol $\mathcal{P}$, is there an execution from an initial configuration
$u$ to a $\top$-stable one?
We use a technique from~\cite{BGKMWW24} to show a general reduction to the syntax problem.

\begin{restatable}{lemma}{lemEmptinessToSyntax}\labelandarrows{lem:emptiness-to-syntax}
  The emptiness problem for \(\PP[\mathcal{N}]\) (\resp \(\IOPP[\mathcal{N}]\)) reduces to deciding the
  complement of the syntax of \(\PP[\mathcal{N}]\) (\resp \(\IOPP[\mathcal{N}]\)) deciders.
\end{restatable}
\begin{proof}[Proof sketch.]
From a protocol $\cP=(Q,\Sigma,O,\Delta)$, we construct a protocol $\cP'$ such that
$\cP'$ is \emph{not} a decider iff $\cP$ can reach a $\top$-stable configuration.
The construction adds a fresh sink state $q_\bot$ with $O(q_\bot)=\bot$, and for every
state $q$ with $O(q)=\bot$ we add a transition allowing an agent in $q$ to switch to $q_\bot$.

As a consequence, from any configuration that contains a $\bot$-agent, a fair execution
can drive the system to the $\bot$-stable consensus $q_\bot^*$.
On the other hand, if a $\top$-stable configuration is reachable in $\cP$, then the
same configuration is reachable in $\cP'$ and remains $\top$-stable there: none of
the added transitions to $q_\bot$ are enabled from $\top$-states.
Therefore $\cP'$ is not a decider.

Conversely, if $\cP$ has no reachable $\top$-stable configuration, then no fair run of
$\cP'$ stabilizes to $\top$.  Every fair run stabilizes
to $q_\bot^*$ and $\cP'$ decides the empty language.
The construction preserves immediate observation, as each new transition updates at
most one agent.
\end{proof}

This reduction, combined with the translation from linear-bounded Turing machines to \(\IOPP[+1]\) from \Cref{sec:beyond}, already yields the following result.

\begin{restatable}{corollary}{corSyntaxSucc}\labelandarrows{cor:syntax:succ}
	The syntax of \(\IOPP[+1]\) and $\PP[+1]$ deciders are undecidable.
\end{restatable}

\begin{restatable}{theorem}{thmSyntaxPP}\labelandarrows{thm:syntax:pp}
  The emptiness problem is undecidable for \(\PP[<]\).  Hence the syntax of
  \(\PP[<]\) deciders is also undecidable.
\end{restatable}

\begin{proof}[Proof sketch.]

	We reduce from the Post correspondence problem.
	Specifically, we take two homomorphisms $h_1, h_2 \colon B^* \to A^*$ and construct a $\PP[<]$ that can reach a $\top$-stable configuration if and only if there is a word $u$ such that $h_1(u) = h_2(u)$.
	We obtain it as a product of three protocols. One makes sure that we can only reach a $\top$-stable configuration from an initial one of the form $\#_1 v \#_2 u \#_3$ with $v \in B^*$ and $u \in A^*$.
	The two others check that we can reach a $\top$-stable configuration only when $h_1(u) = v$ and $h_2(u) = v$, respectively.
	Each one does so by making agents simulate two reading heads going through $u$ and $v$ in lockstep and verifying that $h_i(u) = v$.
	Initially all agents have opinion $\bot$, and in order to change it they must carry the reading head at some point, hence we cannot skip any letter in order to reach a $\top$-stable configuration.
	
	Hence, a $\top$-stable configuration is reachable iff there exist $u, v$ with $h_1(u) = v = h_2(u)$. 
\end{proof}

\subsection{The syntax of \(\IOPP[<]\) deciders is decidable, conditionally}
\label{ssec:syntax:dec}

We present a conjecture on the reachability relation of \(\IOPP[<]\), then show
that it entails decidability of the syntax of their deciders:
\begin{conjecture}\label{conj:DAtoDA}
  The set of configurations reachable \emph{from} a $\DA$ language in an
  \(\IOPP[<]\) protocol is also a $\DA$ language.
\end{conjecture}

\begin{restatable}{theorem}{thmSyntaxIOPP}\labelandarrows{thm:syntax:iopp}
  If \Cref{conj:DAtoDA} holds, then the syntax of \(\IOPP[<]\) deciders is
  decidable.
\end{restatable}

\begin{proof}[Proof sketch.]
	We use two semi-decision procedures.  The first one checks, for each length $n$, if the protocol is a decider on words of length $n$. If the input protocol is not a decider, we will observe it for some $n$.
	The second one looks for invariants witnessing that the protocol is a decider. We look for two disjoint languages $K_\top, K_\bot$ such that every input word is in one of the two, they are closed under transitions of the protocol, and from everywhere in $K_b$ we can reach a $b$-stable configuration.
	Assuming the conjecture, we can assume that those invariants are in \DA.
	We show that we can enumerate potential regular invariants and check the requirements, yielding the second semi-decision procedure.
\end{proof}

In fact, one could slightly strengthen the statement above: it suffices to know that the set of configurations reachable from a $\DA$ language in $\protocol$ (instead of in every \(\IOPP[<]\)-protocol) is also a $\DA$ language, to be able to check whether $\protocol$ is a decider. 
This means that if we can prove \Cref{conj:DAtoDA} over only a subclass of \(\IOPP[<]\)-protocols, then we know that the syntax of that subclass is decidable.

\section{Open questions}

The most tantalizing open question left open by this work is the
characterization of \(\PP[<]\).  We conjecture that the logic and automata models
introduced in \Cref{sec:lbppo} are tight:
\begin{conjecture}\label{con:wuPOPA}
  \(\PP[<]\) is the class of languages recognized by
  \(\Deltaint_1[<, +, \equiv]\) and the class of languages recognized by weakly
  unambiguous poPA.  The regular languages of \(\PP[<]\) are exactly \(\upcom\).
\end{conjecture}

Conjecture~\ref{con:wuPOPA} would entail in particular that weakly unambiguous poPA
are closed under complement.  We note that it is not known whether weakly
unambiguous PA themselves are closed under complement, but it may be easier to
show such closure under the partially ordered assumption.

\Cref{conj:DAtoDA} is also very natural and left open.

\bibliography{references}

\clearpage
\appendix

\section{Appendix}

\subsection{Missing proofs from \Cref{ssec:semidec}}

\lemSemidec*

\begin{proof}
  $\Rightarrow$) This is immediate, since a protocol deciding $L$ is a semi-decider for $L$,
  and the complement protocol, obtained by swapping \(\top\) and \(\bot\) in the opinion
  function, is a semi-decider for $\Sigma^+ \setminus L$.

  $\Leftarrow$) This requires combining the two semi-deciding protocols $\cP_+, \cP_-$, for $L$
  and $\Sigma^+ \setminus L$ respectively, into a protocol $\cP$.  We let \(\cP_s = (Q_s, \Sigma,
  O_s, \Delta_s)\), for \(s \in \{+, -\}\), and we define \(\cP = (Q, \Sigma', O, \Delta)\) as follows.

  We let \(Q = Q_+ \times Q_- \times \{+, -\}\), with $\Sigma' = \set{(a,a,s) \mid a \in \Sigma, s = + \text{ iff } a \in L }$.  We identify \(a \in \Sigma\) with the state \((a, a, s) \in Q\)
  where \(s = +\) iff \(a \in L\).  By this identification, the protocol \(\cP\) recognizes a language over $\Sigma$ and is
  correct on single agent runs by definition (since with a single agent no transition can be
  taken).
  The third component is called the
  \emph{belief} of the agent.  The function \(O\) maps each \((q_+, q_-, s)\) to
  \(O_+(q_+)\) if \(s = +\), and \(\neg O_-(q_-)\) otherwise.  Finally, the transitions
  of \(\cP\) are of two kinds:
  \begin{itemize}
  \item \emph{(Simulation.)} A transition of \(\cP_+\) or \(\cP_-\) can be used, on
    the appropriate component \(Q_+\) or \(Q_-\). The two other components stay the same.  
    
  \item \emph{(Flip.)} If two agents have different beliefs, then one of them can flip its belief.  The 
    \bfsf{true} numerical predicate is used for these transitions.
    An agent in state $(q_+, q_-, s)$ can also flip its belief by itself if $O_s(q_s) = \bot$.
  \end{itemize}
  
  Consider \(u \in \Sigma^+\) of length at least 2.  Assume that \(u \in L\), the case where
  \(u \notin L\) being symmetric.
  Let \(v \in \post^*(u)\). Let $v_+$ and $v_-$ be the projections of $v$ on $Q_+$ and $Q_-$ respectively.
  Observe that for both $s \in \set{+.-}$, $v_s$ is reachable from $u$ in $\cP_s$.
  In consequence, since $u \in L$, there is an execution of $\cP_+$ from $v_+$ to a $\top$-stable $w_+$.
  There is also an execution of $\cP_-$ from $v_-$ to a configuration $w_-$ where at least one agent has opinion $\bot$.
  By following both those executions one after the other, we can go from $v$ to a configuration $w$ in $\cP$, so that the projection of $w$ on $Q_+$ and $Q_-$ are $w_+$ and $w_-$.
  
  \begin{itemize}
  	\item If the agent which has opinion $\bot$ in $w_-$ has opinion $-$, it can switch it to $+$.
  	
  	\item Then, we can flip all the other agents to \(+\), by making them observe this agent. The resulting configuration is  $\top$-stable.
  \end{itemize}
  
  As mentioned, the case \(u \notin L\) is symmetric, showing that \(\cP\) is a decider.
  Furthermore, observe that if both $\cP_+$ and $\cP_-$ are \emph{immediate
  observation} protocols, then so is $\cP$.
\end{proof}

\lemInterSemi*
\begin{proof}
  The proof is standard, and we will give, in \Cref{lem:inter:stab}, a
  slightly more intricate proof with the same ideas, so we omit it here.
\end{proof}

\subsection{Missing proofs from \Cref{ssec:stab}}

\propOrdSemi*

\begin{proof}[Proof of correctness.]
  Note that $\cP$ is input-saving and immediate-observation. Let us
  show that $\cP$ semi-decides $L$ with stabilizing inputs. Let $u \in
  \Sigma^+$ and $u \dynto^* v$. Let $v = w_0 \to w_1 \to \cdots$ be a
  fair run. For the special case of $|v| = 1$, note that each state
  starts with output $\top$ and this remains so as no transition is
  ever enabled. Let us assume that $|v| \geq 2$.

  \emph{Case $\inpt{v} \in L$}. The first rule is permanently disabled
  in $w_0$. Hence, by fairness, the second rule must swap the last
  component of each agent to $\top$. Once this happens, the population
  has reached a $\top$-consensus, and the configuration cannot change
  anymore.

  \emph{Case $\inpt{v} \notin L$}. Recall that the first component of
  each agent is fixed from $w_0$ onwards. Let $i < j$ be positions
  such that $\inpt{v}[i] > \inpt{v}[j]$. Using the second rule, all
  agents can change their opinion to $\top$, and then, using the first
  rule, agent $i$ can change its opinion to $\bot$. By fairness, this
  happens infinitely often, which means that the run does not
  stabilize.
\end{proof}

\lemInterStab*

\begin{proof}
  Let $\cP_1 = (Q_1, \Sigma, \Delta_1, O_1)$ and $\cP_2 = (Q_2,
  \Sigma, \Delta_2, O_2)$ be the protocols that respectively
  $\PP[\cN]$-semi-decide $L_1$ and $L_2$ with stabilizing inputs.

  \medskip
  \textbf{Intersection.} We simply take the product of both protocols;
  simulate them independently with a common input component; and
  output the conjunction of their outputs.

  Let $Q_1 = \Sigma \times R_1$ and $Q_2 = \Sigma \times
  R_2$. Formally, we define $\cP = (Q, \Sigma, \Delta, O)$ by $R
  \defeq R_1 \times R_2$, $Q \defeq \Sigma \times R$, $O((\sigma, q_1,
  q_2)) \defeq O_1((\sigma, q_1)) \land O_2((\sigma, q_2))$, and the
  following rules:
  \begin{align*}
    (\sigma, q_1, r), (\sigma', q_2, s) &\xrightarrow{P} (\sigma, q_3,
    r), (\sigma', q_4, s) && \text{for $(\sigma, q_1), (\sigma', q_2)
      \xrightarrow{P} (\sigma, q_3), (\sigma', q_4) \in \Delta_1$} \\
    &&& \text{and $r, s \in R_2$}, \\
    (\sigma, r, q_1), (\sigma', s, q_2) &\xrightarrow{P} (\sigma, r,
    q_3), (\sigma', s, q_4) && \text{for $(\sigma, q_1), (\sigma',
      q_2) \xrightarrow{P} (\sigma, q_3), (\sigma', q_4) \in
      \Delta_2$} \\
    &&& \text{and $r, s \in R_1$}.
  \end{align*}
  We identify each $\sigma \in \Sigma$ with $(\sigma, \sigma, \sigma)$.

  Note that $\cP$ is input-saving, and immediate-observation if it is
  the case of $\cP_1$ and $\cP_2$. Let us show that $\cP$ semi-decides
  $L \defeq L_1 \cap L_2$ with stabilizing inputs. For every $w \in
  Q^+$, let $\pi_1(w) \in Q_1^+$ be the projection of $w$ onto its
  first two components, and let $\pi_2(w) \in Q_2^+$ be the projection
  of $w$ onto its first and third components. Let $u \in \Sigma^+$ and
  $u \dynto^* v$. Let $v = w_0 \to w_1 \to \cdots$ be a fair run from $v$. By
  definition of $\cP$, the sequence $\pi_i(w_0), \pi_i(w_1), \cdots$
  is a fair run of $\cP_i$. Moreover, $\inpt{v} = \inpt{\pi_1(v)} =
  \inpt{\pi_2(v)}$.

  \emph{Case $\inpt{v} \in L$}. Since $\inpt{v} \in L_1 \cap L_2$, the
  runs of $\cP_1$ and $\cP_2$ eventually stabilize to $\top$-stable
  configurations. Since $O$ is a conjunction, the same holds for the
  run of $\cP$.
  
  \emph{Case $\inpt{v} \notin L$}. Let $i \in \{1, 2\}$ be such that
  $\inpt{v} \notin L_i$. The run of $\cP_i$ visits infinitely many
  configurations containing a state $q \in Q_i$ with $O_i(q) = \bot$.
  Since $O$ is a conjunction, the same holds for the run of
  $\cP$. Furthermore, if $\cP_i$ is a decider, then the run of $\cP_i$
  eventually stabilizes to $\bot$-stable configurations, which implies
  the same for the run of $\cP$.

  \medskip
  \textbf{Union.} We cannot simply take the previous construction and
  redefine the output mapping to $O((\sigma, q_1, q_2)) \defeq
  O_1((\sigma, q_1)) \lor O_2((\sigma, q_2))$. For example, if $\cP_1$
  and $\cP_2$ have two agents with outputs $\top \bot$ and $\bot \top$
  respectively, then their disjunction yields $\top \top$, which is
  incorrect as neither $\cP_1$ nor $\cP_2$ is in a $\top$-consensus.

  Instead, we extend $Q$ with an extra component $\{1, 2\}$ that
  indicates which output should be used, and define $O((\sigma, q_1,
  q_2, i)) \defeq O_i((\sigma, q_i))$. The transitions from $\Delta$
  simply ignore this new component. However, we need to add this
  (immediate-observation) rule:
  \begin{multline*}
    (\sigma, q_1, q_2, i), (\sigma', q_1', q_2', i')
    \xrightarrow{\bfsf{true}} (\sigma, q_1, q_2, 3 - i), (\sigma',
    q_1', q_2', i') \\ \text{for $i \neq i' \lor O_i(q_i) = \bot \lor
      O_i(q_i') = \bot$}.
  \end{multline*}
  This way, as long as agents have not agreed on a common choice, or
  as long as their choice is not in a $\top$-consensus, they may
  change their mind.

  Note that this is correct even for ``deciders'': If $\cP_1$ and
  $\cP_2$ both reach $\bot$-stability, then agents will indefinitely
  change their choice, but nonetheless steadily output $\bot$.
\end{proof}

\lemRenaming*

\begin{proof}  
  Let $\cP = (Q, \Sigma, \Delta, O)$ be the protocol that semi-decides
  $L$ with stabilizing inputs. Since $\cP$ is input-saving, its states
  are of the form $Q = \Sigma \times R$.

  We provide a protocol $\cP'$ where, on input $u \in \Gamma^+$, the
  population internally starts with an arbitrary word $v$ such that
  $u \in f(v)$ and runs $\cP$ on $v$. With luck, it may be the case
  that $v \in L$, but this needs not be the case. Therefore, the
  population may change its choice: If an agent of $\cP'$ has input
  $\gamma$ and encounters a $\bot$-state of $\cP$, then it can change
  its internal input of $\cP$ to any letter $\sigma$ such that $\gamma
  \in f(\sigma)$.

  Let us now describe the protocol formally. For every $\gamma \in
  \Gamma$, let $\sigma_\gamma$ be an arbitrary letter\footnote{If no
  such $\sigma_\gamma$ exists, then we can simply map $\gamma$ to a
  dummy state with opinion $\bot$, as no word containing $\gamma$
  belongs to $f(L)$. That being said, we will never invoke the lemma
  with such an $f$.} of $\Sigma$ such that $\gamma \in
  f(\sigma_\gamma)$. To handle the corner case of populations with a
  single agent, if there is choice such that $\sigma_\gamma \in L$,
  then we take one. We define $\cP' = (Q', \Gamma, \Delta', O')$ as
  follows:
  \begin{itemize}
  \item $Q' \defeq \Gamma \times Q$;

  \item We identify each $\gamma \in \Gamma$ with $(\gamma, \sigma_\gamma)$;

  \item $O'((\gamma, q)) \defeq O(q)$;

  \item The transitions of $\Delta'$ are defined by
    \begin{align*}
      (\gamma, q_1), (\gamma', q_2)
      &\xrightarrow{\mathmakebox[15pt][c]{P}} (\gamma, q_3), (\gamma',
      q_4) && \text{for $\gamma, \gamma' \in \Gamma$ and $(q_1, q_2)
        \xrightarrow{P} (q_3, q_4) \in \Delta$}, \\
      (\gamma, q), (\gamma', q')
      &\xrightarrow{\mathmakebox[15pt][c]{\bfsf{true}}} (\gamma,
      \sigma), (\gamma', q') && \text{for $\gamma, \gamma' \in
        \Gamma$, $q, q' \in Q$ and $\sigma \in \Sigma$ such that} \\
      &&& \text{$\gamma \in f(\sigma)$, and $O(q) = \bot$ or $O(q') =
        \bot$}.
    \end{align*}
  \end{itemize}

  The protocol is input-saving. The second rule is
  immediate-observation. Moreover, if $\cP$ is immediate-observation,
  then it is also the case of the first rule. It remains to show that
  $\cP'$ semi-decides $f(L)$ with stabilizing inputs. For all $w \in
  (Q')^+$, let $\pi(w)$ be the projection of $w$ onto its second
  component. Recall that $\inpt{w}$ is the projection of $w$ onto its
  first component. Let $u \in \Gamma^+$ and $u \dynto^* v$. Let $v =
  w_0 \to w_1 \to \cdots$ be a fair run. The case where $|v| = 1$, and
  hence $\inpt{v} = \gamma \in \Gamma$, is trivially correct by the
  choice of $\sigma_\gamma$. Thus, suppose that $|v| \geq 2$.

  \emph{Case $\inpt{v} \in f(L)$}. By hypothesis, there exists $v' \in
  L$ such that $\inpt{v} \in f(v')$. For the sake of contradiction,
  suppose that no $\top$-stable configuration is visited by the fair
  run. By assumption, there exist indices $i_0 < i_1 < \cdots$ such
  that $O'(w_{i_j}) = \bot$ for all $j \geq 0$. By definition of $O'$,
  this means that $O(\pi(w_{i_j})) = \bot$ for all $j \geq 0$. By
  fairness and the second rule, we can change the second component of
  the population to $v'$ (and hence the input of $\cP$ to $v'$). Since
  $\cP$ semi-decides $L$ with stabilizing inputs, we can reach a
  $\top$-stable configuration of $\cP$ in the second
  component. Consequently, the second rule becomes permanently
  disabled, which implies that $\cP'$ visits a $\top$-stable
  configuration, a contradiction.

  \emph{Case $\inpt{v} \notin f(L)$}. For the sake of contradiction,
  suppose that there exists $j \geq 0$ such that $O'(w_j) =
  O'(w_{j+1}) = \cdots = \top$. This means that $w_j \to w_{j+1} \to
  \cdots$ only uses the first rule. By definition of $O'$, we have
  $O(\pi(w_j)) = O(\pi(w_{j+1})) = \cdots = \top$. Since $\cP$
  semi-decides $L$ with stabilizing inputs, this implies that
  $\inpt{\pi(w_j)} \in L$ and hence $\inpt{w_j} \in f(L)$, which is a
  contradiction.
\end{proof}

\subsection{Missing proofs from \Cref{sec:ioppoinda}}

\propBConsensuses*

\begin{proof}
  In a \(\PP[<]\), we have $u \to u'$ if and only if there exist $a,b,c,d \in
  Q$ and $u_1, u_2, u_3 \in Q^*$ such that $u = u_1 a u_2 b u_3$, $u' = u_1 c u_2 d
  u_3$ and $(a,b) \xrightarrow{<} (c,d) \in \Delta$. Let $w\in Q^*$ be such that $u \preceq w$. We can write $w$ as $w = w_1 a w_2 b w_3$ with $u_1 \preceq w_1$, $u_2 \preceq w_2$ and $u_3 \preceq
  w_3$. We obtain $w \to w' = w_1 c w_2 d w_3$ and $u' \preceq w'$ as desired.
  
  Let $U_0 = Q^* \setminus O^{-1}(b)^*$, the set of configurations which are \emph{not}
  $b$-consensuses.  For all $i$, let $U_{i+1} = U_i \cup \set{u \in Q^* \mid (\exists v \in
  U_i)[u \to v]}$, the set of configurations that can reach $U_0$ in at
  most $i+1$ steps.  We let $D_i = Q^{*} \setminus U_i$.  The set of $b$-stable
  configurations is $D = Q^* \setminus \bigcup_{i \in \N} U_i$.
  
  A key observation is that since our only predicate is $<$, all $D_i$ are
  subword-closed. As a consequence, so is their intersection $D$. Furthermore,
  the sequence $(D_i)_{i \in \N}$ is a descending chain of subword-closed sets.
  Since $\preceq$ is a well quasi-order, every such chain eventually stabilizes,
  hence there exists $i \in \N$ such that $D_i = D_{i+1}$.  We can compute $D_i$ for
  each index until we find one such that $D_i = D_{i+1}$.  Such $D_i$ is then
  the set of $b$-stable configurations.
\end{proof}

  We first show that our protocols are robust to the insertion of
repeated letters, under specific contexts:

\ForwardExtension*
	\begin{proof}
		We prove this for a single step, the result follows immediately by induction.
		Suppose there is a step from $u w v$ to $u' w' v'$, let $(a,b) \xrightarrow{<} (a',b')$ be
		the associated transition.  Since the protocol is IO, we have $a = a'$ or
		$b=b'$.  We assume $a=a'$, the other case being symmetric.
		
		If the observing agent is in $u$ or $v$, then $w= w'$ and we can go from $u w
		z w v$ to $u' w' z' w' v' = u' w z w v'$ with the same transition.  If the
		observing agent is in $w$, then $u = u'$, $v = v'$ and, from $u w z w v$, we
		start by applying that transition for each $b$ in $z$, observing the same
		agent labeled $a$ as in the step from $u w v$ to $u' w' v'$.  We thus turn
		the factor $z$ into $z'$ where all $b$ are turned into $b'$.  We can then
		apply that same transition twice to turn each of the two copies of $w$ to
		$w'$.  Note that $\alpha(z') = (\alpha(z) \cup \set{b'}) \setminus \set{b}$, and $\alpha(w') \supseteq (\alpha(w) \cup
		\set{b'}) \setminus \set{b} \supseteq \alpha(z')$ as $w'$ is obtained from $w$ by deleting an occurrence of $b$
		and adding an occurrence of $b'$.  We obtain $u w' z' w' v$, which is equal to $u' w' z' w'
		v'$ (since here $u=u'$ and $v = v'$) with $\alpha(z') \subseteq \alpha(w')$.
	\end{proof}

\lemBigPump*

\begin{proof}

  
  \Cref{lem:b-consensuses} indicates that for both $b \in \set{\top,\bot}$, the set of
  $b$-stable configurations is subword-closed.  As a consequence, writing
  \(B^\eps\) for \(B \cup \{\eps\}\), it is a finite union of languages of the form
  $A_1^* B_1^\eps \cdots A_k^* B_k^\eps$ with
  $A_i, B_i$ subsets of $Q$, the set of states of our protocol~\cite[Sect.~6.1.1]{Halfon18}.  Let $K_b$ be
  the maximal factorization size $k$ over all those languages,
  $K = \max(K_\top, K_\bot)$ and $m = 2K+1$.
  \Cref{lem:b-consensuses} also states that the set of
  $b$-stable configurations is computable for both $b$, in particular we can compute $K_\top, K_\bot$ and $m$.

  Let $w_1,\dots, w_m, z\in \Sigma^*$ be such that $\alpha(w_1) =\dots =\alpha(w_m) \supseteq\alpha(z)$, and let
  $u, v\in \Sigma^*$.  We show that for both $b \in \{\top, \bot\}$, if we can reach a
  $b$-stable configuration from $u w_1\cdots w _m v$ then it is also possible from $u w_1\cdots w
  _m z w_1\cdots w _m v$. This implies that $u w_1\cdots w _m v$ is accepted if
  and only if $u w_1\cdots w _m z w_1\cdots w _m v$ is, proving the lemma.

  Let $b\in \set{\top, \bot}$, suppose we can reach a $b$-stable configuration $w'$ from
  $u w_1\cdots w _m v$.  Since steps are length-preserving, we can divide $w'$
  into $u' w'_1 \cdots w' _m v'$ where $u'$, $v'$ and each $w'_i$ have the same
  length as $u$, $v$ and each $w_i$, respectively.

  Since $w'$ is $b$-stable, there exists $k < K_b$ and an expression $A_1^*
  B_1^\eps\cdots A_k^* B_k^\eps$ whose language contains $w'$, and only contains
  $b$-stable words.  Since $m> 2 K_b \geq 2k$, there exist $i, j$ such that $u'
  w'_1\cdots w'_{i-1}\in A_1^* B_1^\eps\cdots A_j^*$, $w'_{i+1}\cdots w'_m v'\in A_j^*
  B_j^\eps\cdots B_k^\eps$, and $w'_i\in A_j^*$.

  Since $w = u w_1\cdots w_m v \to^* u' w'_1\cdots w'_m v' = w'$,
  by \Cref{lem:forward-extension} there is a partial run from
  \[
  \widetilde{w} := (u w_1\cdots w_{i-1}) w_i (w_{i+1}\cdots w _m) z (w_{1}\cdots w_{i-1}) w_i (w_{i+1}\cdots w _m v)
  \] to
  \[
  \widetilde{w'} := (u' w'_1\cdots w'_{i-1}) w'_i z' w'_i (w'_{i+1}\cdots w'_m v')
  \] for some $z'$
  with $\alpha(z')\subseteq\alpha(w_i')$, and $|z'| = |(w_{i+1}\cdots w _m) z (w_{1}\cdots w_{i-1})|$.
  
  Furthermore, since we have $u' w'_1\cdots w'_{i-1}\in A_1^* B_1^*\cdots A_j^*$,
  and $w'_{i+1}\cdots w'_m v'\in A_j^* B_j^\eps\cdots B_k^\eps$, and $z'\in\alpha(w_i')^*\subseteq
  A_j^*$, we infer $\widetilde{w'}\in A_1^* B_1^\eps\cdots A_k^* B_k^\eps$, hence $\widetilde{w'}$ is
  $b$-stable.

  We have shown that if there is an execution from $u w_1\cdots w _m v$ to a
  $b$-stable configuration, then there is also one from $u w_1\cdots w _m z
  w_1\cdots w _m v$ to a $b$-stable configuration, for all $u,v\in\Sigma^*$ and both $b\in
  \set{\top, \bot}$.  By definition of the language of $\mathcal{P}$, this means that $w_1\cdots
  w _m$ and $w_1\cdots w _m z w_1\cdots w _m$ are equivalent with respect to $\equiv_L$.
\end{proof}

\lemReg*

\begin{proof}
  We start by showing that some patterns, reminiscent of
  \emph{sesquipowers}~\cite{Lothaire02}, are unavoidable in long strings.  An
  \emph{$m$-reducible pattern} is a word \(w_1 \cdots w_m z w_1 \cdots w_m\) such
  that $\alpha(w_1) = \dots = \alpha(w_m) \supseteq \alpha(z)$.  We have:
  \begin{claim}\label{alphabet-induction}
    For all alphabet $Q$ of size $\ell$ and all $m \in \mathbb{N}$, there exists
    $B(\ell, m) \in \mathbb{N}$ such that for all $w \in Q^*$, if $|w| \geq B(\ell, m)$ then
    $w$ contains an $m$-reducible pattern.
  \end{claim}
    \begin{proof}
      We proceed by induction on the size $\ell$ of the alphabet $Q$.  If $Q$ is a
      singleton, then the result is clear with $B(1,m) = 2m +1$.
    
      Otherwise, define $B(\ell, m) = m B(\ell-1, m) \ell^{m B(\ell-1, m)}$.  Let
      $w \in Q^*$ be such that $|w| \geq B(\ell,m)$, let $u$ be an infix of $w$ of
      maximal length such that $|\alpha(u)| < \ell$, and let $M = |u|$.
      \begin{itemize}
      \item If $M \geq B(\ell-1, m)$, then by induction hypothesis $u$ contains an
        $m$-reducible pattern, thus so does $w$.
      \item Otherwise, we can split $w$ into
        $r = B(\ell, m) / B(\ell-1, m) = m \ell^{m B(\ell-1, m)}$ factors of length
        $B(\ell-1, m)$, plus a suffix $v$: $w = u_1 \cdots u_{r}v$. Since
        $M < B(\ell-1, m)$, by definition of $M$, each $u_k$ contains all letters
        in $Q$.  Since $r \geq m \ell^{m B(\ell-1, m)}$, there must exist indices
        $i,j$ such that $i+m < j$ and
        $u_i \cdots u_{i+m-1} = u_j \cdots u_{j+m-1}$. As all $u_k$ contain
        the same set of letters, we obtain an $m$-reducible pattern
        $(u_i \cdots u_{i+m-1}) (u_{i+m} \cdots u_{j-1}) (u_j \cdots u_{j+m-1})$.\qedhere
      \end{itemize}
    \end{proof}

  By \Cref{alphabet-induction} applied with $m$ taken from
  \Cref{bigpump}, there is a bound $B$ such that every word of length more
  than $B$ is equivalent under $\equiv_L$ to a shorter word.  As a consequence, every
  word is equivalent to a word of length at most $B$.
  
  This means that $\equiv_L$ has finitely many equivalence classes, yielding the
  result by the Myhill-Nerode theorem.
  To construct an automaton for this language from an \(\IOPP[<]\)-protocol, it suffices to have one state for each word of length at most $B$ (computable by \Cref{bigpump} and \Cref{alphabet-induction}). For each of those words $w$ and each letter $a$, there is an $a$-transition from $w$ to $wa$ if it has length $\leq B$. Otherwise, $wa$ has length $B+1$ and contains an $m$-reducible pattern. Let $v$ be the word obtained by reducing it. The $a$-transition from $wa$ goes to $v$. The empty word marks the initial state, and final states are those marked by words accepted by the \(\IOPP[<]\)-protocol, which can be checked by exploring the finite configuration space reachable from them using the protocol.
\end{proof}

\subsection{Missing proofs from \Cref{ssec:da:iopp}}

\propIOEqone*

\begin{proof}
  We define $\cP = (Q, \Sigma, \Delta, O)$ by $Q \defeq \Sigma \times
  \Sigma \times \{\top, \bot\}$, $O((\sigma, \sigma', o)) \defeq (o
  \land \sigma = \sigma')\lor ((\sigma = a \neq \sigma')$, and the
  following rules:
  \begin{align}
    (\sigma, a, o), q &\xrightarrow{\bfsf{true}} (\sigma, \sigma, \bot), q
    & \text{for } \sigma \neq a,\label{rule:eq:1} \\
    (a, \sigma, o), q &\xrightarrow{\bfsf{true}} (a, a, \top),
    q,\label{rule:eq:2} \\[5pt]
    (\sigma, \sigma, o), (a, a, \top) &\xrightarrow{\bfsf{true}}
    (\sigma, \sigma, \top), (a, a, \top)
    & \text{for $\sigma \neq a$},\label{rule:eq:3} \\
    (\sigma, \sigma, o), (\sigma', \sigma', \bot)
    &\xrightarrow{\bfsf{true}} (\sigma, \sigma, \bot), (\sigma',
    \sigma', \bot),\label{rule:eq:4} \\
    (a, a, o), (a, a, o') &\xrightarrow{\bfsf{true}} (a, a, \bot), (a,
    a, o').\label{rule:eq:5}
  \end{align}
  We identify $a$ with $(a, a, \top)$, and each $\sigma \in \Sigma
  \setminus \{a\}$ with $(\sigma, \sigma, \bot)$.

  Note that $\cP$ is input-saving and immediate-observation. Let us
  show that $\cP$ decides the language of the statement with stabilizing inputs. Let $u \in
  \Sigma^+$ and $u \dynto^* v$. Let $v = w_0 \to w_1 \to \cdots$ be a
  fair run.

  If $|v| = 1$, then $v$ is of the form $(a, a, \top)$, $(a, \sigma',
  \bot)$, $(\sigma, a, \top)$, or $(\sigma, \sigma', \bot)$ where
  $\sigma, \sigma' \neq a$. By definition of $O$, the first
  (\resp last) two states have output $\top$ (\resp $\bot$) as
  desired.

  Suppose $|v| \geq 2$. By fairness and
  rules~\eqref{rule:eq:1}--\eqref{rule:eq:2}, we can assume without
  loss of generality that each agent has its first two components
  equal and immutable along the fair run.
  
  \emph{Case $|\inpt{v}|_a = 0$}. Since each $\sigma \neq a$ is
  identified with $(\sigma, \sigma, \bot)$, and since an $a$ can only
  disappear via rule~\eqref{rule:eq:1}, there must exist a position
  $i$ and $\sigma \neq a$ such that $v[i] = (\sigma, \sigma,
  \bot)$. By fairness and rule~\eqref{rule:eq:4}, a $\bot$-consensus
  can be reached. From there, all rules are disabled.

  \emph{Case $|\inpt{v}|_a = 1$}. The single agent whose first
  component is $a$ can set its third component to $\top$ with
  rule~\eqref{rule:eq:2}, and set the third component of other agents
  to $\top$ with rule~\eqref{rule:eq:3}. From there, all rules are
  disabled and hence a $\top$-stable configuration has been reached.

  \emph{Case $|\inpt{v}|_a \geq 2$}. By fairness,
  rules~\eqref{rule:eq:2} and~\eqref{rule:eq:5} cause the third
  component of $a$-agents to alternate indefinitely between $\top$
  and $\bot$.
\end{proof}

\subsection{Missing proofs from \Cref{ssec:fo:int}}

\lemSigNormalForm*

\begin{proof}
  Let $\varphi = (\exists x_1, \ldots,
  x_m)[\psi] \in \Sigmaint_1[<, +, \equiv]$ where $\psi $ is quantifier-free. Without loss of
  generality, $\psi$ may assume that $1 = x_0 < x_1 < \cdots < x_m <
  x_{m+1} = \lmax$. Indeed,
  \begin{itemize}
  \item For the first and last positions, we add the two extra
    variables with these two constraints, and replace each other
    occurrence of $1$ and $\lmax$ by $x_0$ and $x_{m+1}$;

  \item For the ordering, we can change the formula so that it tests
    all orderings:
    \[
    \bigvee_{\pi \in S_m} (\exists x_0, x_1, \ldots, x_m, x_{m+1})[x_1
    \leq \cdots \leq x_m \land \psi_\pi],
    \]
    where $S_m$ is the set of permutations over $[1..m]$, and
    $\psi_\pi$ is $\psi$ with $x_i$ replaced by $x_{\pi(i)}$;

  \item For distinctness, we can similarly enumerate all possible
    equivalence classes of ``$=$'' and associate a variable to each
    class, \eg if we guess that $x_0 < x_1 = x_2 < x_3 = x_4
    < x_5$, then we change the subformula to $(\exists y_0, y_1, y_2,
    y_3)[\psi']$ where $\psi'$ is the formula obtained by renaming
    variables in $\psi$ as follows: $x_0 \mapsto y_0, \{x_1, x_2\}
    \mapsto y_1$, $\{x_3, x_4\} \mapsto y_2$ and $x_5 \mapsto y_3$.
  \end{itemize}

  Once variables are strictly ordered, we can further assume that each
  expression $\cnt{a}(x_i, x_j)$ satisfies $j \in \{i, i +
  1\}$. Indeed, if $i < j$, then $\cnt{\sigma}(x_i, x_j)$ can be
  substituted with
  \[
  \sum_{\ell=i}^{j - 1} \cnt{\sigma}(x_\ell, x_{\ell+1}) -
  \sum_{\ell=i+1}^{j - 1} \cnt{\sigma}(x_\ell, x_\ell).
  \]

  Let $\psi'$ denote the formula $\psi$ where each expression
  $\cnt{\sigma}(x_i, x_{i+1})$ is replaced with $x_{i, \sigma}$, and
  each $\cnt{\sigma}(x_i, x_i)$ is replaced with $y_{i, \sigma}$. Consider $K = \set{a_0 w_1 a_1 \cdots w_n a_m \mid w_i \in \Sigma^\ast, a_i \in \Sigma, \psi'(x_{i, \sigma} \mapsto |w_i|_\sigma, y_{i, \sigma} \mapsto |a_i|_{\sigma})}$. We
  claim that $v \models (\exists x_0, \ldots, x_m)[\psi]$ iff $v \in K$.

  $\Rightarrow$) We take $a_i \defeq v[x_i]$ for each $i \in [0..m]$,
  and $w_i \defeq v[x_{i-1}+1..x_i-1]$ for each $i \in [1..m]$. By the
  assumptions on $\psi$, it is the case that $\psi'(x_{i, \sigma}
  \mapsto |w_i|_\sigma, y_{i, \sigma} \mapsto |a_i|_\sigma)$ holds.

  $\Leftarrow$) We have $v = a_0 w_1 a_1 w_2 \cdots w_m a_m$ where
  $w_i \in \Sigma^*$, $a_i \in \Sigma$ and $\psi'(x_{i, \sigma}
  \mapsto |w_i|_\sigma, y_{i, \sigma} \mapsto |a_i|_\sigma)\}$
  holds. By taking $x_i = |a_0 w_1 \cdots a_i|$ for $i \in [0..m]$,
  $\psi(x_0, \ldots, x_m)$ holds.
\end{proof}

\corDeltaPP*

\begin{proof}
  Let $L \in \Deltaint_1[<, +, \equiv]$. Since $L \in \Sigmaint_1[<,
    +, \equiv]$, by \Cref{prop:sigma:pp}, $L$ is
  $\PP[<]$-semi-decidable. As the complement of $L$ also belongs to
  $\Sigmaint_1[<, +, \equiv]$, the same applies. Thus, we are done by
  \Cref{lem:semidec}.
\end{proof}

\subsection{Missing proofs from \Cref{ssec:popa}}

\propSigPopa*

\begin{proof}
  $\Leftarrow$) Let $L \in \Sigmaint_1[<, +,
    \equiv]$. By \Cref{lem:sig:normal:form}, $L$ can be written as a
  finite union of languages of the form $K = \{a_0 w_1 a_1 \cdots w_m
  a_m \mid w_i \in \Sigma^*, a_i \in \Sigma, \varphi(x_{i, \sigma}
  \mapsto |w_i|_\sigma, y_{i, \sigma} \mapsto |a_i|_\sigma)\}$ where
  $m \geq 0$ and $\varphi$ is a Presburger formula. Since poPA are
  nondeterministic, they are trivially closed under union. Thus, it
  suffices to give a poPA for $K$, which we do below. Here, each
  transition labeled with ``$\sigma$'' stands for $|\Sigma|$ distinct
  transitions:
  
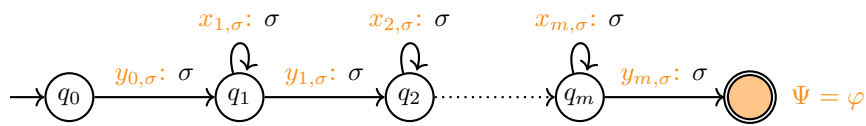
\begin{center}
    \begin{tikzpicture}[->, thick, node distance=2.25cm, initial text=, auto]
  \tikzstyle{astate} = [state, minimum size=18pt, inner sep=1pt];
  \newcommand{\ltrans}[1]{{\color{colA}$#1$:}}
  
  \node[astate, initial]     (q0) {$q_0$};
  \node[astate, right of=q0] (q1) {$q_1$};
  \node[astate, right of=q1] (q2) {$q_2$};
  \node[astate, right of=q2] (qm) {$q_m$};
  \node[astate, right of=qm, accepting, fill=colA!50] (f) {};

  \node[right=2pt of f, colA] {
    $\Psi \defeq \varphi$
  };

  \path[->, dotted]
  (q2) edge node {} (qm)  
  ;
  
  \path[->]
  (q0) edge node {\ltrans{y_{0, \sigma}} $\sigma$} (q1)
  (q1) edge node {\ltrans{y_{1, \sigma}} $\sigma$} (q2)
  (qm) edge node {\ltrans{y_{m, \sigma}} $\sigma$} (f)

  (q1) edge[loop above] node {\ltrans{x_{1, \sigma}} $\sigma$} ()
  (q2) edge[loop above] node {\ltrans{x_{2, \sigma}} $\sigma$} ()
  (qm) edge[loop above] node {\ltrans{x_{m, \sigma}} $\sigma$} ()
  ;
\end{tikzpicture}
  \end{center}

  $\Rightarrow$) Let $\A$ be a poPA. Let $K$ be the words of
  $L(\A)$ of length at most one. We can construct a formula for each
  word of $K$:
  \[
  \varphi_\ew \defeq \bigwedge_{\sigma \in \Sigma} \cnt{\sigma}(1, \lmax) = 0
  \text{ and }
  \varphi_a \defeq \cnt{a}(1, \lmax) = 1 \land
  \bigwedge_{\sigma \neq a} \cnt{\sigma}(1, \lmax) = 0.
  \]

  Let us now consider $K' \defeq L(\A) \setminus K$. As $\A$ is
  partially ordered, its language can be described as the finite union
  of poPA organized as straight lines alternating between
  self-loops and progress transitions. As we consider $K'$ (and so words of length at least 2), we can
  explicitly read the first and last letters, which yields automata of
  the following form, where $a_i \in \Sigma$ and $A_i \subseteq
  \Sigma$:
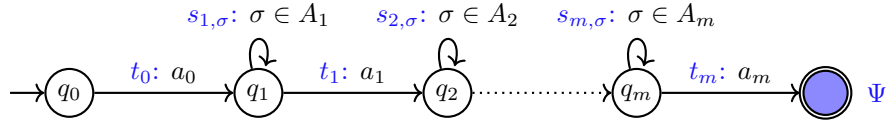
  \begin{center}
    \begin{tikzpicture}[->, thick, node distance=2.5cm, initial text=, auto]
  \tikzstyle{astate} = [state, minimum size=18pt, inner sep=1pt];
  \newcommand{\ltrans}[1]{{\color{colB}$#1$:}}
  
  \node[astate, initial]     (q0) {$q_0$};
  \node[astate, right of=q0] (q1) {$q_1$};
  \node[astate, right of=q1] (q2) {$q_2$};
  \node[astate, right of=q2] (qm) {$q_m$};
  \node[astate, right of=qm, accepting, fill=colB!50] (f) {};

  \node[right=2pt of f, colB] {
    $\Psi$
  };

  \path[->, dotted]
  (q2) edge node {} (qm)  
  ;
  
  \path[->]
  (q0) edge node {\ltrans{t_0} $a_0$} (q1)
  (q1) edge node {\ltrans{t_1} $a_1$} (q2)
  (qm) edge node {\ltrans{t_m} $a_m$} (f)

  (q1) edge[loop above] node {\ltrans{s_{1, \sigma}} $\sigma \in A_1$} ()
  (q2) edge[loop above] node {\ltrans{s_{2, \sigma}} $\sigma \in A_2$} ()
  (qm) edge[loop above] node {\ltrans{s_{m, \sigma}} $\sigma \in A_m$} ()
  ;
\end{tikzpicture}
  \end{center}
  We convert such an automaton into this formula:
  \begin{multline*}
    (\exists x_0, x_1, \ldots, x_m)
    [1 = x_0 < x_1 < \cdots < x_m = \lmax] \land
    \bigwedge_{i=0}^m \cnt{a_i}(x_i, x_i) = 1 \land {} \\[-5pt]
    \bigwedge_{i=1}^m \sum_{\sigma \in \Sigma \setminus A_i}
    \cnt{\sigma}(x_{i-1}+1, x_i-1) = 0 \land 
    \Psi(t_i \mapsto 1, s_{i, \sigma} \mapsto \cnt{\sigma}(x_{i-1}+1, x_i-1)).
  \end{multline*}
  We are done by taking the disjunction of all formulas.
\end{proof}

\obsPaComp*

\begin{proof}
  First, observe that the two automata of \Cref{fig:automata} for the
  median and coDyck-witness languages can be complemented as follows,
  while remaining weakly unambiguous:
  
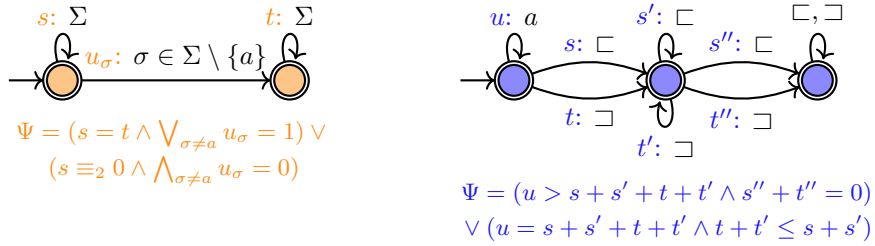
\begin{center}
    \begin{tikzpicture}[->, thick, initial text=, auto]
  \tikzstyle{astate} = [state, minimum size=13pt, inner sep=1pt];
  \newcommand{\ltrans}[1]{{\color{colA}$#1$:}}
  
  \node[astate, initial, accepting, fill=colA!50] (q0) {};
  \node[astate, right=2.5cm of q0, accepting, fill=colA!50] (q1) {};
  
  \node[below=2pt of q0, xshift=42pt, colA] {
    \begin{tabular}{c}
      $\Psi \defeq (s = t \land \bigvee_{\sigma \neq a} u_\sigma = 1) \lor {}$ \\
      $(s \equiv_2 0 \land \bigwedge_{\sigma \neq a} u_\sigma = 0)$
    \end{tabular}
  };
  
  \path[->]
  (q0) edge[loop above] node {\ltrans{s} $\Sigma$} ()
  (q0) edge node {\ltrans{u_\sigma} $\sigma \in \Sigma \setminus \{a\}$} (q1)
  (q1) edge[loop above] node {\ltrans{t} $\Sigma$} ()
  ;

  \renewcommand{\ltrans}[1]{{\color{colB}$#1$:}}
  
  \node[astate, initial, right=70pt of q1, accepting, fill=colB!50] (q0) {};
  \node[astate, right=1.50cm of q0, accepting, fill=colB!50] (q1) {};
  \node[astate, right=1.50cm of q1, accepting, fill=colB!50] (q2) {};
  
  \node[below=25pt of q1, colB] {
    \begin{tabular}{r}
      $\Psi \defeq (u > s + s' + t + t' \land s'' + t'' = 0)$ \\
      ${} \lor (u = s + s' + t + t' \land t + t' \leq s + s')$
    \end{tabular}
  };
  
  \path[->]
  (q0) edge[loop above] node {\ltrans{u} $a$} ()

  (q0) edge[bend  left=20] node[]     {\ltrans{s} $\sqsubset$} (q1)
  (q0) edge[bend right=20] node[swap] {\ltrans{t} $\sqsupset$} (q1)
  
  (q1) edge[loop above] node {\ltrans{s'} $\sqsubset$} ()
  (q1) edge[loop below] node {\ltrans{t'} $\sqsupset$} ()

  (q1) edge[bend  left=20] node[]     {\ltrans{s''} $\sqsubset$} (q2)
  (q1) edge[bend right=20] node[swap] {\ltrans{t''} $\sqsupset$} (q2)
  
  (q2) edge[loop above] node {$\sqsubset, \sqsupset$} ()
  ;
\end{tikzpicture}
  \end{center}
  
  In general, let $\A$ and $\A'$ be poPA with $L(\A') =
  \overline{L(\A)}$. By \Cref{prop:sig:popa}, $L(\A), L(\A') \in
  \Sigma_1[<, +, \equiv]$. Thus, $L(\A) = \overline{L(\A')} \in
  \Pi_1[<, +, \equiv]$ and so $L(\A) \in {\Delta_1[<, +, \equiv]}$.
\end{proof}

\subsection{Missing proofs from \Cref{ssec:regpp}}

\propNotAb*

\begin{proof}
  To derive a contradiction, suppose that $(ab)^+$
  belongs to $\Sigmaint_1[<, +, \equiv]$. By
  \Cref{lem:sig:normal:form}, $(ab)^+ = K_1 \cup \cdots \cup K_\ell$
  where $K_j$ is of the form $\{a_{j,0} w_{j,1} a_{j,1} \cdots
  w_{j,m_j} a_{j,m_j} \mid w_{j,i} \in \Sigma^*, a_{j,i} \in \Sigma,
  \varphi_j(x_{i, \sigma} \mapsto |w_{j,i}|_\sigma, y_{i, \sigma}
  \mapsto |a_{j,i}|_\sigma)\}$, $m_j \geq 0$ and $\varphi_j$ is a
  Presburger formula.

  Let $v \defeq (ab)^{n + 1}$ where $n \defeq \max(m_1, \ldots,
  m_\ell)$. We have $v \in L$ and so $v \in K_j$ for some $j \in
  [1..\ell]$. Let $v = a_{j,0} w_{j,1} a_{j,1} \cdots w_{j,m_j}
  a_{j,m_j}$ be decomposed as in $K_j$. By $|v| = 2n + 2 > 2m_j +
  1$ and the pigeonhole principle, there is $i \in [1..m_j]$ with
  $|w_{j,i}| \geq 2$. Let $w_{j,i}'$ be a word obtained from $w_{j,i}$
  by swapping two adjacent letters. Let $v'$ be obtained from $v$ by
  this change. As the letter counts within $w_{j,i}$ and $w_{j,i}'$
  are the same, we have $v' \in K_j$, a contradiction.
\end{proof}

\subsection{Missing proofs from \Cref{ssec:iopp-pp-succ}}

\lemPPtoIOPPsucc*

\begin{proof}
	Let $\protocol=(Q,\Sigma,O,\Delta)$ be a $\PP[+1]$-protocol.
	We construct a new $\IOPP[+1]$-protocol $\protocol'=(Q',\Sigma,O',\Delta')$ semi-deciding the same language.

        \medskip
	\noindent\emph{States and outputs.}
	For each transition $\delta=(a,b) \xrightarrow{+1} (c,d)$ in $\Delta$, we introduce auxiliary states
	$a^\delta$ and $b^{\mathsf{ack},\delta}$.
	Formally, $Q' \defeq Q \cup \{\, a^\delta,\ b^{\mathsf{ack},\delta}\mid \delta=(a,b) \xrightarrow{+1} (c,d)\in\Delta\,\}$.
	The output function ignores the auxiliary markers: $O'(q)=O(q)$, $O'(a^\delta)=O(a)$ and $O'(b^{\mathsf{ack},\delta})=O(b)$.
		
        \medskip
	\noindent\emph{Transitions.}
        The $\bfsf{true}$-transitions are left unchanged. Furthermore, 
	for each $\delta=(a,b) \xrightarrow{+1} (c,d) \in \Delta$, the set $\Delta'$ contains these four transitions:
	\begin{align*}
		a,b &\xrightarrow{+1} a^\delta,b\\
		a^\delta,b &\xrightarrow{+1}  a^\delta,b^{\mathsf{ack},\delta}\\
		a^\delta,b^{\mathsf{ack},\delta} &\xrightarrow{+1}  c,b^{\mathsf{ack},\delta}\\
		c,b^{\mathsf{ack},\delta} &\xrightarrow{+1}  c,d.
			\end{align*}
			Each of these steps updates at most one of the two agents at a time, and hence is immediate-observation.
			To avoid blocking in the presence of overlapping partial simulations, we also include the following ``clean-up'' rules:
			\begin{align*}
				a^\delta,x &\xrightarrow{+1} a,x && \text{for } x\neq b^{\mathsf{ack},\delta},\\
				y,b^{\mathsf{ack},\delta} &\xrightarrow{+1} y,d && \text{for } y\neq a^\delta.
			\end{align*}

			\noindent\emph{Step simulation.}
			\begin{claim}\label{claim:succ:step-P-to-Pprime}
				If $u,v\in Q^*$ and $u\to v$ in $\protocol$, then $u\to^* v$ in $\protocol'$.
			\end{claim}
			\begin{proof}
				This is trivial for $\bfsf{true}$-transitions. Let $\delta=(a,b)\xrightarrow{+1}(c,d)$ be the transition used in $\protocol$.
				Let $u=u_1ab\,u_2$ and $v=u_1cd\,u_2$.
				In $\protocol'$, the four transitions associated with $\delta$ yield 
				$u_1ab\,u_2\to u_1a^\delta b\,u_2\to u_1a^\delta b^{\mathsf{ack},\delta}\,u_2\to u_1c\,b^{\mathsf{ack},\delta}\,u_2\to u_1cd\,u_2$.
			\end{proof}

			\noindent\emph{Projection to $\protocol$.}
			We define a projection $f \colon Q'^*\to Q^*$ that interprets a $\protocol'$-configuration as a $\protocol$-configuration.
			For $w=w_1\cdots w_n$, define $f(w)=w'_1\cdots w'_n$ by:
			\begin{itemize}
				\item if $w_i\in Q$, then $w'_i=w_i$;
				\item if $w_i=a^\delta$ for some $\delta=(a,b) \xrightarrow{+1} (c,d)$ and $i<n$ with $w_{i+1}=b^{\mathsf{ack},\delta}$, then $w'_i=c$;
				\item if $w_i=a^\delta$ but the above condition fails, then $w'_i=a$;
				\item if $w_i=b^{\mathsf{ack},\delta}$ for some $\delta=(a,b)\to(c,d)$, then $w'_i=d$.
			\end{itemize}
			
			\begin{claim}\label{claim:succ:run-Pprime-to-P}
			  Let $u\in Q^*$ and $v\in Q'^*$. If $u \to^* v$ in $\protocol'$, then $u \to^* f(v)$ in $\protocol$.
			\end{claim}
			\begin{proof}
				We prove the claim by induction on the length of $u \to^* v$ in $\protocol'$.
				The only step that can change the projected configuration is the acknowledgement step
				$(a^\delta,b)\xrightarrow{+1}(a^\delta,b^{\mathsf{ack},\delta})$, in which case $f$ precisely applies the corresponding transition
				$\delta=(a,b)\xrightarrow{+1}(c,d)$ on the projected word. All other rules leave $f(\cdot)$ unchanged.
			\end{proof}
			
			\begin{claim}\label{claim:succ:normalize-to-f}
				For every $v\in Q'^*$ reachable in $\protocol'$, it is the case that $v \to^* f(v)$ in $\protocol'$.
			\end{claim}
			\begin{proof}
				Starting from $v$, repeatedly apply the clean-up rules to remove every pending marker $a^\delta$
				not followed by the matching $b^{\mathsf{ack},\delta}$, and then complete each matching pair
				$a^\delta b^{\mathsf{ack},\delta}$ by the two remaining simulation steps
				$(a^\delta,b^{\mathsf{ack},\delta}) \xrightarrow{+1} (c,b^{\mathsf{ack},\delta}) \xrightarrow{+1} (c,d)$.
				This yields $f(v)$.
			\end{proof}
			
			We are now ready to prove the lemma. Let us show that $w \in L(\protocol)$ iff $w \in L(\protocol')$. 
			
			 $\Rightarrow$) Let $w \to^* u$ in $\protocol'$. Thanks to \Cref{claim:succ:run-Pprime-to-P}, we know that $w \to^* f(u)$ in $\protocol$. Since $w \in L(\protocol)$, configuration $f(u)$ can reach a $\top$-stable configuration $v$ in $\protocol$. Hence, thanks to \Cref{claim:succ:normalize-to-f} and \Cref{claim:succ:step-P-to-Pprime}, we conclude that $u \to^* f(u) \to^* v$ in $\protocol'$. To derive a contradiction, suppose that $v$ is not $\top$-stable in $\protocol'$. Configuration $v$ can reach a non-$\top$-consensus $v'$ in $\protocol'$.
			Consider $v'$ to be such a configuration reachable in a minimal number of steps. The last step must be a transition of the form $(a^{\delta}, b^{\textsf{ack}, \delta}) \xrightarrow{+1} (c, b^{\textsf{ack}, \delta})$ with $O'(c)=O(c)=\bot$, or 
			$(c, b^{\textsf{ack}, \delta}) \xrightarrow{+1} (c, d)$, with  $O'(d)=O(d)=\bot$.
			Hence, $f(v')$ is not a $\top$-consensus in $\protocol$.
			Therefore, $f(v) = v \to^* f(v')$ in $\protocol$ by \Cref{claim:succ:run-Pprime-to-P}, contradicting the $\top$-stability of $v$ in $\protocol$.			
			
			$\Leftarrow$) Let $w \to^* u$ in $\protocol$. 
			By \Cref{claim:succ:step-P-to-Pprime}, $w \to^* u$ in $\protocol'$. As $w \in L(\protocol')$, configuration $u$ can reach a $\top$-stable configuration $v$ in $\protocol'$. From \Cref{claim:succ:normalize-to-f}, we have $v \to^* f(v)$ in $\protocol'$. Note that $f(v)$ is $\top$-stable in $\protocol'$ by construction. Moreover, from \Cref{claim:succ:run-Pprime-to-P}, we have $u \to^* f(v)$ in $\protocol$. To derive contradiction, suppose that $f(v)$ is not $\top$-stable in $\protocol$. We have $f(v) \to^* v'$ in $\protocol$ for some non-$\top$-consensus $v'$. Thanks to \Cref{claim:succ:step-P-to-Pprime}, we obtain $f(v) \to^* v'$ in $\protocol'$, contradicting the $\top$-stability of $f(v)$ in $\protocol'$. Hence, $f(v)$ is $\top$-stable in $\protocol$, and so $w \in L(\protocol)$.
\end{proof}

\newcommand{\qstate}[1]{\mathsf{sta}(#1)}
\newcommand{\qpos}[1]{\mathsf{pos}(#1)}

The rest of the subsection is dedicated to proving the following
lemma.

\lemPPsuccSemiDecidesLTM*

Let us fix a linear-bounded nondeterministic Turing machine $M = (P,
\Sigma, \Sigma', \delta, p_0, p_\text{acc}, p_\text{rej})$ where $P$
is the set of states; $\Sigma$ is the input alphabet; $\Sigma'
\supseteq \Sigma$ is the tape alphabet; $\delta \subseteq (P \times
\Sigma') \times (P \times \Sigma' \times \{\Tleft, \Tright\})$ is the
transition relation; and $p_0, p_\text{acc}, p_\text{rej}$ are
respectively the initial, accepting and rejecting states.

Traditionally, on input $w \in \Sigma$, a linear-bounded automaton has
access to $\lmarker w \rmarker$; has no left-move on $\lmarker$; has
no right-move on $\rmarker$; and cannot overwrite the endmarkers. For
convenience, we consider instead that there are no endmarkers, but
that the head stays put when moving left (resp.\ right) from the first
(resp.\ last) cell. It is simple to construct such a machine with the
same language. We further assume that $M$ may only enter
$p_\text{acc}$ with its head on the rightmost cell, and that there is
no transition from $p_\text{acc}$.

Let $\Gamma \defeq (P \cup \{-\}) \times \Sigma' \times \{\text{fst},
-, \text{lst}\}$. For every $s = (p, a, x) \in \Gamma$, let
$\qstate{s} \defeq p$ and $\qpos{s} \defeq x$. We say that a word $w
\in \Gamma^+$, with $|w| \geq 2$, is a \emph{valid configuration} of
$M$ if
\begin{itemize}  
\item $\qpos{w[i]} = \text{fst}$ iff $i = 1$,

\item $\qpos{w[i]} = \text{lst}$ iff $i = |w|$,

\item there is a unique position $i$ with $\qstate{w[i]} \in P$.
\end{itemize}

For example, $w = (-, a, \text{fst}) (p, a, -) (-, b, \text{lst})$
represents the configuration where $M$ is in state $p$, with its head
on the second cell, and its tape containing $aab$.

\begin{proposition}\label{prop:simulate:tm}
  Let $$L \defeq \{w \in \Gamma^{\geq 2} \mid \text{$w$ is a valid
    configuration of $M$ and it leads to acceptance}\}.$$ One can
  construct $\PP[+1]$-protocol that semi-decides $L$ with stabilizing
  inputs. Furthermore, its $\bfsf{true}$-transitions are
  immediate-observations.
\end{proposition}

\newcommand{\sRes}{\text{\raisebox{1pt}{\scriptsize\faUserCog}}}

\begin{proof}
  Let us construct a protocol $\cP = (Q, \Gamma, O, \Delta)$. We
  define $Q \defeq \Gamma \times \Gamma \times \Gamma \times \{\bot,
  \top, \sRes\}$ and $O((x, y, z, o)) \defeq (x = y \land o =
  \top)$. We identify each input $\gamma \in \Gamma$ with $(\gamma,
  \gamma, \gamma, \bot)$.

  Intuitively, the state $(x, y, z, o)$ indicates that the agent has
  input $x$; has been acting as if its input was $y$; currently holds
  $z$; and has output belief $o$. Whenever $x \neq y$, the agent has
  changed its mind on its input, and must consequently reset the whole
  population.

  The transitions of $\Delta$ are defined by these rules:
  \begin{center}
    \begingroup
    \setlength{\tabcolsep}{2pt}
    \resizebox{\textwidth}{!}{%
    \begin{tabular}{crclp{0pt}l}
      \toprule
      & \multicolumn{3}{c}{\textbf{\emph{Simulation of $M$}}} \\
      \midrule
      (1) &
      $(\ast, \ast, (p, a, \ast), \ast), (\ast, \ast, (-, \ast, \ast), \ast)$
      &
      $\xrightarrow{+1}$ &
      $(\ast, \ast, (-, b, \ast), \ast), (\ast, \ast, (q, \ast, \ast), \ast)$
      &&
      for $((p, a), (q, b, \Tright)) \in \delta$ \\
      (2) &
      $(\ast, \ast, (-, \ast, \ast), \ast), (\ast, \ast, (p, a, \ast), \ast)$
      &
      $\xrightarrow{+1}$ &
      $(\ast, \ast, (q, \ast, \ast), \ast), (\ast, \ast, (-, b, \ast), \ast)$
      &&
      for $((p, a), (q, b, \Tleft)) \in \delta$ \\
      (3) &
      $(x, \ast, (p, a, \ast), \ast)$
      &
      $\xrightarrow{\bfsf{true}}$ &
      $(x, \ast, (q, b, \ast), \ast)$
      &&
      for $((p, a), (q, b, \Tright)) \in \delta$ and $\qpos{x} = \text{lst}$ \\
      (4) &
      $(x, \ast, (p, a, \ast), \ast)$
      &
      $\xrightarrow{\bfsf{true}}$ &
      $(x, \ast, (q, b, \ast), \ast)$
      &&
      for $((p, a), (q, b, \Tleft)) \in \delta$ and $\qpos{x} = \text{fst}$ \\
      \midrule
      & \multicolumn{3}{c}{\textbf{\emph{Belief propagation}}} \\
      \midrule
      (5) &
      $(x, \ast, \ast, \ast), (x', \ast, z', \ast)$ &
      $\xrightarrow{\bfsf{true}}$ &
      $(x, \ast, \ast, \top), (x', \ast, z', \ast)$ &&
      if $\qpos{x} = \text{fst} \land \qpos{x'} = \text{lst} \land
      \qstate{z'} = p_\text{acc}$ \\
      (6) &
      $(\ast, \ast, \ast, \top), (\ast, \ast, \ast, \ast)$ &
      $\xrightarrow{+1}$ &
      $(\ast, \ast, \ast, \top), (\ast, \ast, \ast, \top)$ \\
      \midrule
      & \multicolumn{3}{c}{\textbf{\emph{Reset on input changes}}} \\
      \midrule
      (7) &
      $(x, y, \ast, \ast), (x', y', \ast, \ast)$ &
      $\xrightarrow{\bfsf{true}}$ &
      $(x, y, \ast, \ast), (x', y', \ast, \sRes)$ &&
      if $\qpos{x'} = \text{lst} \land (x \neq y \lor x' \neq y')$ \\
      (8) &
      $(\ast, \ast, \ast, \ast), (x', \ast, \ast, \sRes)$ &
      $\xrightarrow{+1}$ &
      $(\ast, \ast, \ast, \sRes), (x', x', x', \bot)$ \\
      (9) &
      $(x, \ast, \ast, \sRes)$ & $\xrightarrow{\bfsf{true}}$ &
      $(x, x, x, \bot)$ && if $\qpos{x} = \text{fst}$ \\
      \midrule
      & \multicolumn{3}{c}{\textbf{\emph{Erroneous configuration detection}}} \\
      \midrule
      (10) &
      $(x, \ast, \ast, \ast), (x', \ast, \ast, \ast)$ &
      $\xrightarrow{\bfsf{true}}$ &
      $(x, \ast, \ast, \ast), (x', \ast, \ast, \bot)$ &&
      if
      $\qstate{x}, \qstate{x'} \in P$ \\
      (11) &
      $(x, \ast, \ast, \ast), (x', \ast, \ast, \ast)$ &
      $\xrightarrow{\bfsf{true}}$ &
      $(x, \ast, \ast, \ast), (x', \ast, \ast, \bot)$ &&
      if $\qpos{x} = \qpos{x'} \in \{\text{fst}, \text{lst}\}$ \\
      (12) &
      $(x, \ast, \ast, \ast), (x', \ast, \ast, \ast)$ &
      $\xrightarrow{+1}$ &
      $(x, \ast, \ast, \bot), (x', \ast, \ast, \bot)$ &&
      if $\qpos{x} = \text{lst} \lor \qpos{x'} = \text{fst}$ \\
      \bottomrule
    \end{tabular}}
    \endgroup
  \end{center}

  Note that all $\bfsf{true}$-transitions are
  immediate-observation. Let us show that $\cP$ semi-decides $L$ with
  stabilizing inputs. Let $u \in \Gamma^n$. For the special case of $n
  = 1$, note that each state starts with output $\bot$ and this
  remains so as no transition is ever enabled. Thus, let us assume
  that $n \geq 2$.

  Let $u = u_0 \dynto u_1 \dynto \cdots \dynto u_\ell = v$, and let $v
  = w_0 \to w_1 \to \cdots$ be a fair run. Without loss of generality,
  we may assume that $\ell = 0$ or $\inpt{u_\ell} \neq
  \inpt{u_{\ell-1}}$. Indeed, it suffices to take $\ell$ as the last
  moment where the input changes.

  We make a case distinction, where each case assumes that the
  previous ones do not hold.
  
  \medskip
  \noindent\emph{Case~1: $\inpt{v}$ is a valid configuration.} By
  validity, rules~(10--12) are permanently disabled from $v$ onwards.

  If $\ell > 0$, then there exists a position $i$ such that $v[i] =
  (x, y, z, o)$ with $x \neq y$. By validity of $\inpt{v}$, we have
  $\qpos{\inpt{v}[1]} = \text{fst}$ and $\qpos{\inpt{v}[n]} =
  \text{lst}$. Thus, by fairness, the population must eventually use
  rule~(7) to assign \sRes\ to agent $n$; use rule~(8) to reset the
  population from right to left; and use rule~(9) to complete this
  reset with agent~$1$. From there, rules~(7--9) are permanently
  disabled. Thus, by fairness, $M$ is simulated faithfully with
  rules~(1--4), and $\top$ eventually spreads with rules~(5--6).

  The case of $\ell = 0$ is the same except for the fact that
  rules~(7--9) are disabled from the very beginning, where the
  simulation of $M$ already begins.

  \medskip
  \noindent\emph{Case~2: $\inpt{v}$ has several ``fst'', ``lst'' or
  states from $P$; or a wrongly placed ``fst'' or ``lst''.} Since the
  first component of each agent stops changing from $v$ onwards,
  fairness and rules~(10--12) guarantee that at least one $\bot$-belief
  must occur infinitely often.

  \medskip
  \noindent\emph{Case~3: $\ell = 0$ and $\inpt{v}$ is invalid.}
  Rule~(5) is never enabled and hence no $\top$ ever appears.
  
  \medskip
  \noindent\emph{Case~4: $\ell > 0$ and $\inpt{v}$ is invalid.} By
  assumption, there is a position $i$ such that $v[i] = (x, y, z, o)$
  with $x \neq y$. By the latter, we have $O(v[i]) = \bot$. To obtain
  a $\top$-consensus, agent $i$ must be reset. This requires having at
  least one \sRes\ in the population. To get rid of all \sRes,
  rule~(9) must eventually be used. If this happens, then a
  ``fst''-agent now has belief $\bot$. By assumption, this agent must
  be the first one. The only way to set its belief to $\top$ is by
  using rule~(5), which is impossible by assumption on $\inpt{v}$.
\end{proof}

\begin{proposition}
  The language $L(M)$ is $\IOPP[+1]$-semi-decidable.
\end{proposition}

\begin{proof}
  Let $K$ denote the set of valid initial configurations of
  $M$. Formally, let $K \defeq \Gamma_0 \Gamma_1^* \Gamma_2$ where
  $\Gamma_0 \defeq \{p_0\} \times \Sigma \times \{\text{fst}\}$,
  $\Gamma_1 \defeq \{-\} \times \Sigma \times \{-\}$ and $\Gamma_2
  \defeq \{-\} \times \Sigma \times \{\text{lst}\}$.

  Let $f(i) \defeq \Gamma_i$ and $U_i \defeq \{w \in \{0, 1, 2\}^+ :
  |w|_i = 1\}$. We have
  \[
  K = f(0^* 1^* 2^* \cap U_0 \cap U_2).
  \]

  By \Cref{prop:ord:semi}, the language $0^* 1^* 2^*$ is
  $\IOPP[<]$-semi-decidable with stabilizing inputs. Replacing ``$<$''
  by ``$+1$'' in its protocol directly yields
  $\IOPP[+1]$-semi-decidability with stabilizing inputs. Thus, by
  \Cref{prop:io:eqone} and by \Cref{lem:renaming,lem:inter:stab}, the
  language $K$ is $\IOPP[+1]$-semi-decidable with stabilizing inputs.

  Let $K' \defeq \{w \in \Gamma^{\geq 2} \mid \text{$w$ is a valid
    configuration of $M$ and it leads to acceptance}\}$. By
  \Cref{prop:simulate:tm}, there is a $\PP[+1]$-protocol that
  semi-decides $K'$ with stabilizing inputs and whose
  $\bfsf{true}$-transitions are immediate-observations. Let $g((x, y,
  z)) \defeq y$. We have
  \[
  L(M) \cap \Sigma^{\geq 2} = g(K \cap K').
  \]
  Thus, by \Cref{lem:inter:stab,lem:renaming}, the language $L(M) \cap
  \Sigma^{\geq 2}$ is $\PP[+1]$-semi-decidable. Moreover, these two
  lemmas preserve immediate-observation
  $\bfsf{true}$-transitions. Hence, by \Cref{lem:PP-to-IOPP-succ}, the
  language $L(M) \cap \Sigma^{\geq 2}$ is $\IOPP[+1]$-semi-decidable.

  We can trivially show that $L(M) \cap \Sigma$ is
  $\IOPP[\emptyset]$-decidable. Thus, we are done by taking the union,
  with \Cref{lem:inter:semi}.
\end{proof}

\subsection{Missing proofs from \Cref{ssec:syntax:undec}}

\lemEmptinessToSyntax*
\begin{proof}
	Take a protocol $\mathcal{P} = (Q, \Sigma, O, \Delta)$.  We build a protocol $\mathcal{P}'$ such that $\mathcal{P}'$ is \emph{not}
	a decider if and only if a $\top$-stable configuration is reachable in $\mathcal{P}$.
	
	If $\Sigma$ contains a letter $a$ with opinion $\top$, then the configuration $a$ is $\top$-stable and reachable. In that case, we set $\cP'$ as an arbitrary protocol which is not a decider.
	
	Otherwise, we build $\mathcal{P}' = (Q', \overline{\Sigma}, O', \Delta')$ from $\cP$ as follows: For each initial state $a \in \Sigma$ in $\cP$, we
	add another state $\overline{a}$ with opinion $\bot$. We denote $\overline{\Sigma}$
	the set of those additional states. 
	We also add another state $q_\bot$ with opinion $\bot$.
	The set of initial states of $\cP'$ is $\overline{\Sigma}$.  
	For each $\overline{a}\in \overline{\Sigma}$ and $q \in Q'$ we add a transition $(\overline{a},q) \xrightarrow{\bfsf{true}} (a,q)$.
	As a consequence, from every initial configuration $\overline{a_1} \cdots \overline{a_k}$ with $k\geq 2$, we can reach its counterpart in $a_1 \cdots a_k$ in $\Sigma^*$.
	
	Further, for every pair of states $q_1, q_2$ such that at least one of the two has opinion $\bot$, we add a transition $(q_1, q_2) \xrightarrow{\bfsf{true}} (q_1, q_\bot)$.
	Note that by this last transition, from every configuration of length $\geq 2$ which is not a $\top$-consensus, we can reach a configuration of ${q_\bot}^*$.

	Suppose that there is a run in $\cP$ from an initial configuration $u = a_1
	\cdots a_k \in \Sigma^+$ to a $\top$-stable configuration $v$ in $\cP$. We must have $k \geq 2$
	since we assumed that all $a \in \Sigma$ have opinion $\bot$.
	In $\cP'$, the configuration
	$\overline{u} = \overline{a}_1 \cdots \overline{a}_k$ is initial, and it can reach both ${q_\bot}^k$ and $v$,
	respectively a $\bot$ and a $\top$-stable configuration.  As a result, the protocol is not
	a decider.
	
	Now suppose that a $\top$-stable configuration is unreachable in $\cP$.  Let $u$ be an
	initial configuration and let $u \to^* v$. Since $v$ is not 
	$\top$-stable, we have $v \to^* w$ for some configuration $w$
	containing a state with opinion $\bot$. Thus $w \to^* q_\bot^*$ and the latter is a $\bot$-consensus.  As a result, the protocol is a decider
	(and recognizes the empty language).
	
	Finally, observe that if $\cP$ is immediate-observation then so is $\cP'$.
\end{proof}

\corSyntaxSucc*

\begin{proof}
	In \Cref{sec:beyond}, we presented an effective construction translating a linear-bounded Turing machine into an $\IOPP[+1]$ protocol with the same language.
	Since emptiness is undecidable for linear-bounded Turing machines~\cite[Thm.~5.10]{Si06}, it must also be undecidable for $\IOPP[+1]$ protocols.
	By \Cref{lem:emptiness-to-syntax}, so is the syntax of \(\IOPP[+1]\) deciders.
\end{proof}

\thmSyntaxPP*

\begin{proof}
  We reduce from the Post correspondence problem (PCP), which is undecidable.  We are given two finite
  alphabets $A, B$, and two homomorphisms $h_1, h_2 \colon B^* \to A^*$, and must determine
  whether there exists $w \in B^*$ such that $h_1(w) = h_2(w)$.
  
  We construct three protocols, each in charge of verifying a property of the
  initial word:
  \begin{itemize}
  \item $\cP_0$ checks that the input word is in $\#_1 A^* \#_2 B^* \#_3$;
    
  \item $\cP_1$ checks that $h_1(u) =v$ for a word $\#_1 v \#_2 u \#_3$;
    
  \item $\cP_2$ checks that $h_2(u) =v$ for a word $\#_1 v \#_2 u \#_3$.
  \end{itemize}

  For the first one, we simply use the fact that $\#_1 A^* \#_2 B^* \#_3$ is in
  $\DA$, and hence recognized by an (immediate-observation) protocol
  $\cP_0$.  In particular an input word can reach a $\top$-stable configuration if and only if
  it is in that language.
  
  We now construct $\cP_i$ for $i \in \set{1,2}$.  The set of states of both protocol is
$$
    \set{a, a^{head}, a^{next} \mid a \in A}\\ \cup~  B \cup\set{(b, \pi) \mid b \in B, \pi \text{
    prefix of } h(b)}\\ \cup~  \set{\#_1, \#_2, \overline{\#_2}, \#_3, q_\top}.
  $$The opinion of $q_\top$ is $\top$, all other states have opinion $\bot$.
  
  Protocol rules are as follows, for all $a,c\in A$, $b,d\in B$ and prefix $\pi$ of
  $h_i(b)$:
  \begin{center}
    \begingroup
    \setlength{\tabcolsep}{2pt}
    \begin{tabular}{rclp{25pt}rclp{5pt}l}
      \toprule
      \multicolumn{3}{c}{\textbf{\emph{Start}}} &&
      \multicolumn{3}{c}{\textbf{\emph{Head}}} \\
      \midrule
      $\#_1, a$ & $\xrightarrow{<}$ & $q_{\top}, a^{head}$ &&
      ${a}^{head}, (b, \pi)$ & $\xrightarrow{<}$ & $a^{next}, (b, \pi a)$ \\
      $\#_2, b$ & $\xrightarrow{<}$ & $\overline{\#_2}, (b, \eps)$
      && && \multicolumn{3}{r}{if $\pi a$ is a prefix of $h_i(b)$} \\
      \midrule
      \multicolumn{3}{c}{\textbf{\emph{Next (A)}}} &&
      \multicolumn{3}{c}{\textbf{\emph{Next (B)}}} \\
      \midrule
      ${a}^{next}, c$ & $\xrightarrow{<}$ & $q_{\top}, c^{head}$ &&
      $(b, h_i(b)), d$ & $\xrightarrow{<}$ & $q_{\top}, (d, \eps)$ \\
      ${a}^{next}, \overline{\#_2}$ & $\xrightarrow{<}$ & $q_{\top}, q_{\top}$ &&
      $(b, h_i(b)), \#_3$ & $\xrightarrow{<}$ & $q_{\top}, q_{\top}$ \\
      \bottomrule
    \end{tabular}
    \endgroup
  \end{center}
  
  Our final protocol $\cP$ is simply the product of those three protocols, the
  opinion of a triple of states being the conjunction of their opinions.  This
  way, we can reach a $\top$-stable configuration from an initial configuration if and only if
  we can reach a $\top$-stable configuration from it in all three protocols.
  
  We now show that this is the case if and only if that initial configuration is of
  the form $\#_1 u \#_2 v \#_3$ with $u = h_1(v) = h_2(v)$.
  
  If the initial configuration $\#_1 h(a_1) \cdots h(a_k) \#_2 a_1 \cdots a_k \#_3$
  satisfies those conditions, then it is straightforward to build an execution
  from it to $q_\top^*$ in both $\cP_1$ and $\cP_2$.  Furthermore, it can reach a
  $\top$-stable configuration in $\cP_0$ since it is in $\#_1 A^* \#_2 B^* \#_3$.
  
  Now consider an initial configuration from which we can reach a $\top$-consensus
  in all three protocols. By definition of $\cP_0$, it must be of the form $\#_1
  a_1 \cdots a_k \#_2 b_1 \cdots b_\ell \#_3$ with $a_1, \dots, a_k \in A$ and $b_1, \dots,
  b_\ell \in B$.
  
  Now consider protocol $\cP_i$, with $i \in \set{1,2}$. We call the states
  $\set{a^{head}, a^{next} \mid a \in A} \cup \set{\#_1}$ \emph{special states}.  It is
  easy to see that while we do not use the last \textbf{\emph{Next (A)}} rule,
  there is exactly one special state in the configuration, and there is no
  special state afterwards.  Also observe that only letters to the right of the
  special state can be modified.  Since all agents must reach $q_\top$, this
  imposes that all agents that have a state in $A$ initially must hold the
  special state, one by one from left to right.
  
  A symmetric argument shows that all agents with a state in $B$ initially must
  transform into $q_\top$ using \textbf{\emph{Next (B)}} rules, one by one, from
  left to right.  Now let us have a look at the \textbf{\emph{Head}} rule:
  $a^{head}, (b, \pi) \xrightarrow{<} a^{next}, (b, \pi a)$. Let $w$ be the word obtained by
  taking the letter $a$ associated with each of those transitions, in the order
  in which they appear in the run.
  
  By the arguments above, we must have $w = a_1 \cdots a_k$, and $w = h_i(b_1)
  \cdots h_i(b_\ell)$.  As a result, the input word is indeed of the form $\#_1
  h_i(b_1) \cdots h_i(b_\ell) \#_2 b_1 \cdots b_\ell \#_3$.
  
  As we have shown this for both $i \in \set{1,2}$, we obtain that the initial
  configuration must be of the form $\#_1 u \#_2 v \#_3$ with $u = h_1(v) =
  h_2(v)$.
  
  In conclusion, there is a reachable $\top$-stable configuration if and only if the PCP
  instance has a solution.
\end{proof}

\subsection{Missing proofs from \Cref{ssec:syntax:dec}}

\thmSyntaxIOPP*

\begin{proof}
  Let \(\post(u) \defeq \{v \mid u \to v\}\) and \(\post^*(u) \defeq
  \{v \mid u \to^* v\}\). Similarly, let \(\pre(u) \defeq \{v \mid v
  \to u\}\) and \(\pre^*(u) \defeq \{v \mid v \to^* u\}\). We extend
  these to sets, \eg $\post^*(L) \defeq \bigcup_{w \in L}
  \post^*(w)$.

  We reuse the machinery of \Cref{sec:ioppoinda}, specifically
  \Cref{lem:b-consensuses}, \Cref{bigpump}, and the combinatorial
  \Cref{alphabet-induction}.

  If a $\IOPP[<]$ protocol $\cP = (Q,\Sigma, O, \Delta)$ is not a decider, then it is witnessed by the
  configuration space of words of length $n$ for some $n \geq 2$. This can be checked
  for each $n$ incrementally.  We only have to provide a semi-decision procedure
  which returns yes if and only if the protocol is a decider.

  Observe that assuming \Cref{conj:DAtoDA} holds, a protocol $\cP$ is a decider if and only if there exist languages $K_\top, K_\bot$
  in $\DA$ over $Q^*$ such that:

  \begin{itemize}
  \item $K_\top$ and $K_\bot$ are disjoint;
  \item $\Sigma^* \subseteq K_\top \cup K_\bot$;
  \item For both $b$, for all $u \in K_b$ and $v \in \post(u)$, it is the case that $v \in K_b$;
  \item For both $b$, $K_b \subseteq \pre^*(C_b)$ with $C_b$ the set of $b$-stable configurations.
  \end{itemize}

  The right-to-left direction is clear: if such invariants exist, then the protocol recognizes $K_\top\cap \Sigma^*$. For the other direction, by \Cref{thm:iopp:da}, if $\cP$ is a decider then its language $L$ is in \DA, hence so is $\Sigma^* \setminus L$. 
  Take $K_\top = \post^*(L)$ and $K_\bot =
  \post^*(\Sigma^* \setminus L)$.
  Further, assuming \Cref{conj:DAtoDA}, $K_\top$ and $K_\bot$ are also in \DA.
  They clearly satisfy all the conditions.

  Given $\cP$, we enumerate automata $\cal A, \mathcal{B}_\top, \mathcal{B}_\bot$ recognizing languages in
  $\DA$.

  We check whether:
  \begin{itemize}
  \item $L(\mathcal{A}) \subseteq L(\mathcal{B}_\top)$, $\Sigma^* \setminus L(\cal A) \subseteq L(\mathcal{B}_\bot)$;
    
  \item For both $b$, for all $u \in L(\mathcal{B}_b)$ and $v \in \post(u)$, it is the case that $v \in L(\mathcal{B}_b)$;
    
  \item For both $b$, $L(\mathcal{B}_b) \subseteq \pre^*(C_b)$.
  \end{itemize}

  The first condition is easy to check.
  For the second one, we can build a non-deterministic automaton recognizing $\bigcup_{u \in L(\mathcal{B}_b)} \post(u)$, by making it guess two positions at which it applies a transition of the protocol, while checking that the resulting word is in $L(\mathcal{B}_b)$.
  The last condition is more subtle. 
   To check it, we argue that if there is a word in $L(\mathcal{B}_b) \setminus \pre^*(C_b)$ then there is  one of bounded length. It then suffices to check inclusion over words of length below this bound.

	First we prove a property of language of \DA by a classical use of Ramsey's theorem.

	\begin{claim}
		\label{claim:Ramsey}
		For every language $L \subseteq Q^*$ in \DA, there exists $m \in \N$, computable from an automaton recognizing $L$, such that
		for all $n \geq m$ and $u_1, \dots, u_n, z$ with $\emptyset \neq \alpha(u_1) = \dots =
		\alpha(u_n) \supseteq \alpha(z)$, we have $u_1\cdots u_n \equiv_{L} u_1 \cdots u_n z u_1 \cdots u_n$.
	\end{claim}

	\begin{claimproof}
		Let $(M,\cdot, 1_M)$ be the syntactic monoid of $L$, $\varphi \colon Q^* \to M$ a morphism and $F \subseteq M$ such that $L = \varphi^{-1}(F)$.
		
		Given a sequence of words $u_1, \dots, u_n$, consider the complete undirected graph (without loops) over $\set{0,\dots, n}$, where for all $i<j$ the edge $\set{i,j}$ is colored $\varphi(u_{i+1} \cdots u_j)$.
		By Ramsey's theorem~\cite{Ramsey}, there exists a uniform bound $m \in \N$ such that whenever $n \geq m$ this graph contains a monochromatic $3$-clique, \ie there are $i<j<k$ such that $\varphi(u_{i+1} \cdots u_j) = \varphi(u_{j+1} \cdots u_k) = \varphi(u_{i+1} \cdots u_k) $.

		Let $u_1, \dots, u_n, z$ with $n\geq m$ and $\emptyset \neq \alpha(u_1) = \dots =
		\alpha(u_n) \supseteq \alpha(z)$, and let $i<j<k$ be as described above.
		Since $\varphi$ is a morphism, $\varphi(u_{i+1} \cdots u_k) = \varphi(u_{i+1} \cdots u_j)\varphi(u_{j+1} \cdots u_k) = \varphi(u_{i+1} \cdots u_k)^2$.
		In other words, $u_{i+1} \cdots u_k \equiv_L (u_{i+1} \cdots u_k)^2$.
		By~\Cref{eq:charac:da}, we have $ u_{i+1} \cdots u_k \equiv_L u_{i+1} \cdots u_k (u_{k+1} \cdots u_n z u_1 \cdots u_i) u_{i+1} \cdots u_k$.
		
		By appending a prefix $u_1 \cdots u_i$ and a suffix $u_{k+1} \cdots u_n$ to both sides of the equivalence, we obtain \(u_1 \cdots u_n \equiv_L (u_1 \cdots u_n)z (u_1 \cdots u_n)\).
	\end{claimproof}

	We are now ready to show the bound on witnesses for (the negation of) the third condition.

	\begin{claim}
		We can compute a bound $k$ such that if  $L(\mathcal{B}_b) \setminus \pre^*(C_b)$ is non-empty then it contains a word of length at most $k$.
	\end{claim}

	\begin{claimproof}
		Since $L(\mathcal{B}_b)$ belongs to $\DA$, we can compute $m$ from $\mathcal{B}_b$ as given by \Cref{claim:Ramsey}. 
		
		By \Cref{lem:b-consensuses}, the set of $b$-stable configurations is subword-closed
		and computable. There is therefore a computable $k$ such that it is made of a finite
		union of languages described by expressions of the form $A_1^* B_1^\epsilon \cdots A_\ell^* B_\ell^\epsilon$ with  $\ell\leq k$.  Let $n =
		\max(m,k)$.
		
		Let $w$ be a word in $L(\mathcal{B}_b) \setminus \pre^*(C_b)$ of length $> B(|\Sigma|, n)$ as given by
		\Cref{alphabet-induction}. The word $w$ is of the form $x u_1 \cdots u_n
		z u_1 \cdots u_n y$ with $\emptyset \neq \alpha(u_1) = \dots = \alpha(u_n) \supseteq \alpha(z) $.
		
		The word $w' = x u_1 \cdots u_n y$ belongs to $L(\mathcal{B}_b)$ since $n \geq m$. It is not in
		$\pre^*(C_b)$ since otherwise $w$ would be in $\pre^*(C_b)$, analogously to the
		proof of \Cref{bigpump}.
	\end{claimproof}

	In conclusion, to check the third condition we can compute the bound $k$ given by the claim above, and check if there is a word of length at most $k$ in  $L(\mathcal{B}_b) \setminus \pre^*(C_b)$.
\end{proof}


\end{document}